\documentclass[11pt]{article}

\usepackage[margin=1.125in]{geometry}
\usepackage{xcolor}
\definecolor{DarkGreen}{rgb}{0.1,0.5,0.1}
\definecolor{DarkRed}{rgb}{0.5,0.1,0.1}
\definecolor{DarkBlue}{rgb}{0.1,0.1,0.5}

\usepackage[small]{caption}
\usepackage[pdftex]{hyperref}
\hypersetup{
    unicode=false,         
    pdftoolbar=true,       
    pdfmenubar=true,       
    pdffitwindow=false,     
    pdfnewwindow=true,     
    colorlinks=true,      
    linkcolor=DarkBlue,         
    citecolor=DarkGreen,       
    filecolor=DarkGreen,     
    urlcolor=DarkBlue,         
    pdftitle={},
    pdfauthor={},    
}
\usepackage{amssymb,amsmath,amsthm,amsfonts,enumitem,bm}
\usepackage{fullpage,nicefrac,comment}
\usepackage{tikz}
\usetikzlibrary{arrows,decorations.pathmorphing,decorations.shapes,snakes,patterns,positioning}
\usepackage[ruled,linesnumbered]{algorithm2e}

\newcommand{\jmin}{{j_{\min}}}
\newcommand{\jmax}{{j_{\max}}}
\newcommand{\ellmin}{\ell_{\min}}
\newcommand{\ellmax}{\ell_{\max}}
\newcommand{\lmin}{\ell_{\min}}
\newcommand{\lmax}{\ell_{\max}}
\newcommand{\dr}{d^{(r)}}
\newcommand{\yr}{y^{(r)}}
\newcommand{\lmaxr}{\ell_{\max}^{(r)}}
\newcommand{\lminr}{\ell_{\min}^{(r)}}
\newcommand{\jmaxr}{j_{\max}^{(r)}}
\newcommand{\jminr}{j_{\min}^{(r)}}

\newcommand{\cI}{\ensuremath{\mathcal{I}}}
\newcommand{\cC}{\ensuremath{\mathcal{C}}}
\newcommand{\cF}{\ensuremath{\mathcal{F}}}
\newcommand{\cV}{\ensuremath{\mathcal{V}}}
\newcommand{\cW}{\ensuremath{\mathcal{W}}}

\newcommand{\cP}{\ensuremath{\mathcal{P}}}

\newcommand{\RR}{{\mathbb R}}
\newcommand{\Z}{{\mathbb Z}}

\newcommand{\F}{{\mathbb F}}
\newcommand{\B}{{\mathbb B}}

\newcommand{\bb}{\mathbf{b}}
\newcommand{\bff}{\mathbf{f}}
\newcommand{\bg}{\mathbf{g}} 
\newcommand{\bv}{\mathbf{v}}
\newcommand{\bw}{\mathbf{w}}
\newcommand{\bu}{\mathbf{u}}
\newcommand{\bc}{\mathbf{c}}
\newcommand{\bp}{\mathbf{p}}

\newcommand{\bx}{\mathbf{x}}
\newcommand{\by}{\mathbf{y}}
\newcommand{\bz}{\mathbf{z}}
\newcommand{\bX}{\mathbf{X}}

\newcommand{\bG}{\ensuremath{\mathbf{G}}}
\newcommand{\btheta}{\bm{\theta}}

\newcommand{\inabs}[1]{\left|#1\right|}

\newcommand{\rflor}{\right\rfloor}
\newcommand{\lflor}{\left\lfloor}
\newcommand{\inset}[1]{\left\{#1\right\}}

\newcommand{\inparen}[1]{\left(#1\right)}

\newcommand{\inbrak}[1]{\left[#1\right]}
\newcommand{\suchthat}{\,:\,}

\newcommand{\degset}{\mathrm{degSet}}
\newcommand{\modstar}{\ \ \mathrm{mod}^*\ }
\newcommand{\supp}{\mathrm{Supp}}

\newcommand{\spn}{\ensuremath{\operatorname{span}}}

\newcommand{\tr}{\mathrm{tr}}

\newcommand{\ind}[1]{\ensuremath{\mathbf{1}_{#1}}}

\newcommand{\eps}{\varepsilon}
\renewcommand{\epsilon}{\varepsilon}
\newcommand{\vphi}{\varphi}

\newtheorem{theorem}{Theorem} 
\newtheorem{lemma}[theorem]{Lemma} 
\newtheorem{definition}[theorem]{Definition}

\newtheorem{observation}[theorem]{Observation}

\newtheorem{corollary}[theorem]{Corollary} 

\newtheorem{remark}{Remark}
\newtheorem{claim}[theorem]{Claim}
\newtheorem{proposition}[theorem]{Proposition}

\newtheorem{fact}[theorem]{Fact}

\usepackage{cancel}

\title{Low-bandwidth recovery of linear functions of Reed-Solomon-encoded data}
\author{Noah Shutty\thanks{Stanford University.  Email: noaj@stanford.edu. N.S. was supported in part by NSF DGE-1656518.} \ and Mary Wootters\thanks{Stanford University.  Email: marykw@stanford.edu.  Research partially supported by NSF Grants CCF-1844628 and CCF-BSF-1814629, and by a Sloan Research Fellowship.}}
\date{\today}
\begin{document}
\maketitle

\begin{abstract}
We study the problem of efficiently computing on encoded data.  More specifically, we study the question of \em low-bandwidth \em computation of functions $F:\F^k \to \F$ of some data $\bx \in \F^k$, given access to an encoding $\bc \in \F^n$ of $\bx$ under an error correcting code.  In our model---relevant in distributed storage, distributed computation and secret sharing---each symbol of $\bc$ is held by a different party, and we aim to minimize the total amount of information downloaded from each party in order to compute $F(\bx)$.  Special cases of this problem have arisen in several domains, and we believe that it is fruitful to study this problem in generality.

Our main result is a low-bandwidth scheme to compute linear functions for Reed-Solomon codes, even in the presence of erasures.  More precisely, let $\eps > 0$ and
let $\cC: \F^k \to \F^n$ be a full-length Reed-Solomon code of rate $1 - \eps$ over a field $\F$ with constant characteristic.
For any $\gamma \in [0, \eps)$, our scheme can compute any linear function $F(\bx)$ given access to any $(1 - \gamma)$-fraction of the symbols of $\cC(\bx)$, with download bandwidth $O(n/(\eps - \gamma))$ bits.  In contrast, the naive scheme that involves reconstructing the data $\bx$ and then computing $F(\bx)$ uses $\Theta(n \log n)$ bits.  Our scheme has applications in distributed storage, coded computation, and homomorphic secret sharing.
\end{abstract}

\section{Introduction}\label{sec:intro}
Suppose that we would like to store some data $\bx \in \F^k$ on a distributed storage system consisting of $n$ nodes, where $n \geq k$.  (Here and for the rest of the paper, $\F$ denotes some finite field).
Since node failure is a possibility, we may protect the data with an \emph{error correcting code} $\cC:\F^k \to \F^n$ as follows. 
We encode $\bx$ as a \emph{codeword} $\bc = \cC(\bx) \in \F^n$, and for $i=1,\ldots,n$, we send the symbol $c_i$ to the $i$'th storage node. 
If $\cC$ is a \emph{Maximum Distance Separable} (MDS) code---meaning that any $k$ symbols of the codeword $\bc$ are sufficient to recover the original data $\bx$---then the system can tolerate $n-k$ node failures without losing any of the original data.
Encoding with an MDS code (such as a Reed-Solomon code, see Definition~\ref{def:rs} below) is common in distributed storage: for example, Reed-Solomon codes are built into HDFS~\cite{HDFS} and Ceph~\cite{ceph}.

Given data $\bx \in \F^k$ encoded and stored with an MDS code as described above, suppose that we would like to compute a function $F(\bx)$ of the data, where $F:\F^k \to \F$.
One scheme (which we will refer to as the {\it naive scheme}) is to contact any $k$ of the nodes, download their data, recover $\bx$, and compute $F(\bx)$.
This requires downloading $k$ field symbols, or $k \log|\F|$ bits. We call the amount of downloaded information the {\it bandwidth} of the scheme.
Given that $F(\bx)$ is only one field symbol, or $\log|\F|$ bits, the naive scheme seems wasteful in terms of bandwidth.

Our motivating question is whether we can compute $F(\bx)$ with less bandwidth.
That is, \emph{when is it possible to do communication-efficient computation on top of encoded data?}

In this paper, we introduce a new notion, \emph{low-bandwidth function evaluation}, in order to make this question precise.  
Our main result is a low-bandwidth function evaluation scheme for the ubiquitous family of \emph{Reed-Solomon Codes}, and for the useful family of \emph{linear functions} $F:\F^k \to \F.$

\subsection{Low-Bandwidth Function Evaluation}\label{sec:problem}
A low-bandwidth evaluation scheme for a code $\cC$ and a collection of functions $\cF$ allows us to compute functions in $\cF$ in a communication-efficient way on data encoded with $\cC$, even when a set $\cI \subset [n]$ of symbols are unavailable (e.g., the corresponding nodes have failed).  More precisely, we have the following definition.
Below, and throughout the paper, we use bold letters like $\bc$ to denote vectors, and we use $c_i$ to denote the $i$'th entry of $\bc$.

\begin{definition}[Low-Bandwidth Function Evaluation]\label{def:repair}
Let $\cC:\F^k \to \F^n$ be a code.  Let $\cF$ be a class of functions $F: \F^k \to \F$.  
Let $b \geq 0$. 
We say that there is an
\emph{evaluation scheme for $\mathcal{F}$ and $\mathcal{C}$ with bandwidth $b$} if
 for any $F \in \cF$, there are:
\begin{itemize}
\item positive integers $b_1, \ldots, b_n \in \Z^{\geq 0}$ so that $\sum_j b_j \leq b$;
\item functions $g_1, \ldots, g_n$ so that $g_j: \F \to \{0,1\}^{b_j}$;
\item and a function $G: \{0,1\}^{\sum_j b_j} \to \F$ 
\end{itemize} 
so that for all $\bx \in \F^k$, if $\bc = \cC(\bx)$, then
\[ G(g_1(c_1), g_2(c_2), \ldots, g_n(c_n)) = F(\bx). \]
We denote the scheme by $\Phi:F \mapsto (g_1, g_2, \ldots, g_n, G)$ that maps $F \in \cF$ to the maps $g_j$ and $G$.

If there is a set $\cI \subset [n]$ so that $b_j = 0$ for all $F \in \cF$ and for all $j \in \cI$, we say that $\Phi$ \emph{tolerates failures in $\cI$}.
\end{definition}

\begin{remark}[More general alphabets]\label{rem:alpha}
More generally, one could define an evaluation scheme for codes $\cC:\Sigma_0^k \to \Sigma_1^n$ for arbitrary input/output alphabets.  In this paper, we focus on linear functions and MDS codes, so we state Definition~\ref{def:repair} with $\Sigma_0 = \Sigma_1 = \F$ being some finite field.
\end{remark}

\begin{remark}[Knowledge of $\mathcal{I}$]\label{rem:I}
We note that in Definition~\ref{def:repair}, the set of failed nodes $\mathcal{I}$ tolerated by a scheme is a property of that particular scheme; a stronger definition might demand that the same scheme tolerates \emph{any} set of failed nodes of a particular size.  In the distributed storage example above, this weaker definition means that the nodes may need to know which nodes have failed in order to decide which scheme to use.  This mirrors the set-up in regenerating codes, discussed below, where the identity of the (single) failed node is assumed to be known.
\end{remark}

Notions related to Definition~\ref{def:repair} have been studied before, for particular families of functions and/or particular codes.  We mention a few of these below, and discuss them more in Section~\ref{sec:apps} (Applications of our results) and Section~\ref{sec:related} (Related work).

\begin{itemize}
\item \textbf{Regenerating codes.}  In the model of distributed storage described above, there has been a great deal of work on \em regenerating codes, \em which aim to repair one node failure with low download bandwidth (see, e.g.,~\cite{dimakis_survey}).  This is a special case of Definition~\ref{def:repair} when $\cF$ is the family of functions $f_i(\bx) = \cC(\bx)_i$, for $i = 1, \ldots, n$, and where $\cI = \{i\}$.  We note that if the code is systematic, this allows us to recover the dictator functions $f_i(\bx) = x_i$.
\item \textbf{Gradient Coding.}  The goal of \em gradient coding \em is to speed up distributed gradient descent in the presence of \em stragglers\em, that is, compute nodes that may be slow or unresponsive~\cite{gradient_coding}.  In this model, the data is $\bX = (\bx^{(1)}, \ldots, \bx^{(k)})$ where each $\bx^{(i)} \in \RR^d$ for some $d$.  The data $\bX$ is distributed using a code among $n$ workers, so that worker $i$ receives $\bc^{(i)} \in \RR^{d'}$, for a codeword $\mathbf{C} = (\bc^{(1)}, \ldots, \bc^{(n)})$.
At each timestep the parameter server (PS) has an iterate $\btheta \in \RR^d$, which it broadcasts to the workers.
Each worker $i$ that has not failed computes a local function $g_i(\bc^{(i)})$ and returns it to the PS.  The PS then uses these messages to recover the gradient of some loss function, $\nabla \mathcal{L}(\mathbf{X};\btheta)$.  
One goal of gradient coding is to tolerate stragglers in any set $\cI$ of some fixed size, while minimizing the communication bandwidth from the workers to the PS (e.g., \cite{YA18}).
This can be cast as a special case of a strengthing of Definition~\ref{def:repair} (with different input/output alphabets as per Remark~\ref{rem:alpha} and which can tolerate any small set $\mathcal{I}$ of failed nodes as per Remark~\ref{rem:I}), where $\cF$ is the family of functions given by possible gradients: $\cF = \inset{ F_{\btheta} \suchthat \btheta \in \RR^d }$, where $F_{\btheta}(\bX) = \nabla \mathcal{L}(\mathbf{X};\btheta)$.
\item \textbf{Homomorphic Secret Sharing.}  In \emph{secret sharing}, a secret $s$ is shared among $n$ parties, so that some coalitions of parties can recover the secret while others learn nothing about it.
A classic example is \em Shamir's scheme, \em which is essentially a Reed-Solomon (RS) code: let $\cC$ be a RS code of dimension $k$ and length $n+1$.  To share a secret $s \in \F$, 
we pick a random codeword $\bc \in \cC$ such that $c_0 = s$, and we send $c_i$ to party $i$.  Any $k$ parties can recover $\bc$ and hence $s$, but any $k-1$ parties learn nothing about $s$.
In {\it single-client} \emph{Homomorphic Secret Sharing} (HSS) \cite{HSS0, HSS1,HSS2}, one additionally asks that the parties be able to locally compute messages $g_i(c_i)$ so that a referee can compute a function $f(s)$ of the secret from these messages.\footnote{Typically, $f(s)$ should be a sum of the messages, in which case the HSS scheme is called \em additive.\em}
The HSS property has applications in Private Information Retrieval and Secure Multiparty Computation (see, e.g., \cite{HSS1}).  In some applications, it is desirable that the messages $g_i(c_i)$ be short, in which case the HSS scheme is said to be \em compact. \em  Low-bandwidth function evaluation is related to (information-theoretic, not-necessarily-additive) compact HSS, in the sense that a
 low-bandwidth function evaluation scheme for a Reed-Solomon code $\cC$ and with $\mathcal{F} = \{F:\bx \mapsto f(\cC(\bx)_0)\}$ gives a single-client HSS protocol for Shamir's scheme; the bandwidth of the scheme corresponds to the compactness of the messages $g_i(c_i)$.  More generally, if $\cC$ represents a secret sharing scheme, then a low-bandwidth function evaluation scheme for $\cC$ yields a compact single-client HSS protocol for that scheme.  
\end{itemize}

Given the numerous places that notions related to Definition~\ref{def:repair} have appeared,
we believe it will be fruitful to study Definition~\ref{def:repair} in generality.  
In this paper we begin this general study by 
considering what is arguably most natural class of functions (after the indicator functions $f(\bx) = \cC(\bx)_i$ studied in regenerating codes): the class of linear functions.
Our main results, described next in Section~\ref{sec:results}, are low-bandwidth evaluation schemes for Reed-Solomon codes, for classes of linear functions.  In Section~\ref{sec:apps}, we mention several applications of our results.

\subsection{Our Results}\label{sec:results}
Our main results hold for Reed-Solomon codes, defined below.
\begin{definition}[Reed-Solomon Code]\label{def:rs}
Let $\mathbb{F}$ be a finite field and suppose $n \leq |\F|$.
Let $\alpha_1, \alpha_2, \ldots, \alpha_n \in \mathbb{F}$ be distinct evaluation points.
The \textbf{Reed-Solomon Code} (RS) of dimension $k$ and length $n$ with evaluation points $\alpha_1, \ldots, \alpha_n$ is the map $\cC:\F^k \to \F^n$ given by
\[ \cC(\bff) = (f(\alpha_1), \ldots, f(\alpha_n)), \]
where for $\bff \in \F^k$, we define $f(X) = \sum_{i=0}^{k-1} f_i X^i$.
\end{definition}

Our contributions are as follows.

\begin{enumerate}
\item \textbf{A framework for computing linear functions on RS-encoded data.}
We provide a framework for developing low-bandwidth evaluation schemes for RS codes and for families $\cF$ of linear functions over extension fields.  This begins with a general linear-algebraic charactization (similar to the characterization of \cite{GW17} for regenerating codes) that applies to \em any \em linear code.
However, we go beyond that, building on it to develop a framework for RS codes in particular.
The linear-algebraic characterization for any linear code is given in Section~\ref{sec:prelim}, and the framework for RS codes is given in Section~\ref{sec:framework}.
\item \textbf{Low-bandwidth schemes for computing any linear function on RS-encoded data, up to the Singleton bound.}
Our main theorem, Theorem~\ref{thm:main}, can be summarized/simplified as follows:
\begin{theorem}[Simplified; see Theorem~\ref{thm:main}]\label{thm:informal}
Let $\eps > \gamma > 0$.
There is some $q_0 = \Theta\inparen{\frac{1}{\eps - \gamma}}$ so that the following holds for sufficiently large $n$ and for any prime power $q \geq q_0$.

Let $Q = q^t$ for any $t \geq 2$. Let $n=Q$ and let $k = (1 - \eps)n$.  Let $\cC$ be the RS code over $\F = \F_Q$ of dimension $k$ and length $n$, with evaluation points all of $\F$.  Let $\cF$ be the set of all linear functions $F:\F^k \to \F$.  Let $\cI \subseteq [n]$ be any set with $|\cI| < \gamma n$.  Then there is an evaluation scheme for $\cC$ and $\cF$ that tolerates failures in $\cI$, with bandwidth (measured in bits) of
\[ b = O\inparen{ \frac{n \log(q)}{\eps - \gamma} }. \]
\end{theorem}
We make a few remarks about Theorem~\ref{thm:informal}:
\begin{itemize}
\item The naive scheme (downloading enough information to recover $\bx$, and then computing $F(\bx)$) requires 
\[ k \log(Q) =  \frac{n \log n}{1 - \eps} \] bits of bandwidth.  
Thus, when $\eps, \gamma, q$ are constant, our scheme gives an asymptotic improvement of a factor of $\log n$ over the naive scheme.  (Notice that $|\mathbb{F}| = Q = q^t$, so we may choose $q$ to be constant and allow $n = Q$ to grow by growing $t$).
\item Our scheme can tolerate a $\gamma$ fraction of failures, where $\gamma$ can be arbitrarily close to $\eps$.  Since the rate of the code is $1 - \eps$, by the Singleton bound the relative distance can be at most $\eps$, and so this is optimal. 
\item One may wonder about lower bounds on the bandwidth.  In Appendix~\ref{app:lb} (Observation~\ref{obs:nonmaxlb}), we show that $b \geq n \log_q\inparen{ \frac{n}{n-k+1}} \approx n \log_q(1/\eps)$ is necessary for Reed-Solomon codes, and that a similar result (Corollary~\ref{cor:lbrs}) holds for any MDS code. This shows that the linear dependence on $n$ in Theorem~\ref{thm:informal} is optimal for constant $\eps$ and $q$, although we leave it as an open question to pin down the correct dependence on $\eps$ and $\gamma$.  
\end{itemize}

We also give a simpler version of Theorem~\ref{thm:main}, Theorem~\ref{thm:ratehalf}, which does not tolerate any failures and which only works up the rate $1/2$.  Our main reason for presenting Theorem~\ref{thm:ratehalf} is for exposition, as the proof is simpler, but we are also able to get more precise constants.  In particular, we show that RS codes of rate approaching $1/2$ (for sufficiently large $q$, where $Q = q^t$ is the size of $\F$) can compute any linear function with bandwidth at most $n \lceil \log q \rceil$.    When $q=2$, we get an RS code of rate $1/4$ with bandwidth $n$ bits, or \emph{only one bit} from each node.

\item \textbf{Applications in distributed storage, coded computation, and homomorphic secret sharing.}  Our results have applications in several domains.  
We elaborate on these next in Section~\ref{sec:apps}.

\end{enumerate}

\subsection{Applications}\label{sec:apps}
As noted above, low-bandwidth function evaluation shows up in several settings, and our work has natural applications in these areas.  We briefly mention a few potential applications of Theorem~\ref{thm:informal} to further motivate our results.  First, we make two remarks about the generality of our scheme.

\begin{remark}[Non-linear functions]\label{rem:nonlinear}
Our framework can also be used to efficiently compute certain non-linear functions, for example $\bx \mapsto \sum_{i=1}^k x_i^2$.  To see this, we first suppose without loss of generality that the Reed-Solomon code has a systematic encoding, so that $\bx$ is encoded as $(f(\alpha_1), \ldots, f(\alpha_n))$ where $f \in \F[X]$ is the unique polynomial of degree at most $k-1$ so that $f(\alpha_i) = x_i$ for $i=1,\ldots,k$.  Then define $g(X) = f(X)^2$ and observe that each node $i=1, \ldots, n$ can locally compute $g(\alpha_i)$.  Thus, we can apply our scheme to the Reed-Solomon code of dimension $2k-1$ to recover the linear combination $\sum_{i=1}^k g(\alpha_i) = \sum_{i=1}^k x_i^2$.
\end{remark}

\begin{remark}[Prime fields]\label{rem:prime}
Our approach requires that $\F$ be an extension field over a base field $\B$.  However, in many applications (including those discussed below), it is desirable to work over a prime field.  The reason is that often we actually want to work over the reals or the integers, and these can be nicely embedded in $\mathbb{F}_p$ for a large enough prime $p$.  Fortunately, for certain linear functions $F$, our approach can still be used to save bandwidth when we wish to amortize several computations over prime fields.  

In more detail, suppose that $\B = \F_p$ for a large prime $p$, and let $\F = \F_{p^t}$.  Let $\zeta_1, \ldots, \zeta_t$ be a basis for $\F$ over $\B$.  Suppose that the linear function we want to compute is $F(\bx) = \bb^T \bx$, where $\bb \in \B^k$ has coefficients in the base field $\B$.  This is the case, for example, in Remark~\ref{rem:nonlinear} when we want to compute the $\ell_2$ norm: all of the coefficients are $1$.  It is also the case when the data $\bx$ represents a histogram and we'd like to take the sum of certain buckets: all of the coefficients are $0$ or $1$.
If $\bb \in \B^k$, then we can proceed as follows.  View the data $\bx \in \F^k$ as $t$ data points in $\B^k$.  That is, we write $x_j = \sum_{i=1}^t y_{i,j} \zeta_i$, and interpret $\bx$ as $t$ vectors $\by^{(i)} = ( y_{i,1}, y_{i,2}, \ldots, y_{i,k} )$.  If we use our scheme to compute $F(\bx)$, then we have computed
\[ F(\bx) = \sum_{j=1}^k b_j x_j = \sum_{j=1}^k b_j \inparen{ \sum_{i=1}^t \zeta_i y_{i,j} } = \sum_{i=1}^t \zeta_i \bb^T \by^{(i)} = \sum_{i=1}^t \zeta_i F(\by^{(i)}). \]
Since $F(\by^{(i)}) \in \B$, and since the $\zeta_i$ form a basis for $\F$ over $\B$, we can now read off the values $F(\by^{(i)})$. 
This allows us to compute $t$ evaluations of $F$ on vectors $\by^{(1)}, \ldots, \by^{(t)} \in \B^k$, using bandwidth $O(n \log p)$ (assuming that the rate $1 - \eps$ of the code and fraction $\gamma < \eps$ of failed nodes are constants).  In contrast, the naive computation would require $O(tn \log p)$ bits.  So for such linear functions $F$, our scheme can do $t = \log_p n$ computations for the bandwidth cost of a single computation in the naive scheme.
\end{remark}

\paragraph{Distributed Storage.} The application to distributed storage was described in Section~\ref{sec:intro}.  In this context, Theorem~\ref{thm:informal} gives a method to compute any linear function of data stored on a distributed storage system with non-trivial download bandwidth.  The reader may be wondering about the \em upload \em bandwidth: don't we need to communicate the function $F$ to each node?  The reason that we focus on the download bandwidth (as is also the case for regenerating codes) is because of the way that files are stored in a typical distributed storage system.  In more detail, a large file $\bx$ will be broken up into blocks $\bx^{(1)}, \ldots, \bx^{(M)} \in \F^k$, where $M$ is very large, and each $\bx^{(i)}$ will be encoded as $\bc^{(i)}$, so that the $j$'th node stores $\{c^{(i)}_j\suchthat j \in [M]\}$.  With this set-up, the evaluation scheme of Theorem~\ref{thm:informal} would be run independently on each of the blocks, so that the upload cost is just the cost of broadcasting $F$, while the download cost is $M$ times the bandwidth guaranteed in the theorem.  Since $M$ is large, the download cost dominates the upload cost, and Theorem~\ref{thm:informal} yields real bandwidth savings over the naive scheme.

\paragraph{Coded Computation and Low-Bandwidth Matrix-Vector Multiplication.} Suppose we would like to distribute some data $\bX$ among $n$ worker nodes and perform a computation $f(\bX)$ in a distributed way. 
A body of work~\cite{LLPPR18, DCG19, YLRKSA19} on \em coded computation \em has proposed introducing redundancy in the data assignment, with the goal of tolerating stragglers (worker nodes that may be slow or non-responsive): that is, we would like responses from any $\tilde{k}$ out of $n$ workers to determine $f(\bX)$.  
There are two lines of work in coded computation.  One line of work adds redundancy by replicating and appropriately distributing data (for example the work on gradient coding mentioned above~\cite{YA18}, or a line of work aimed at general MapReduce computations~\cite{LMYA17}), and aims to minimize download bandwidth.  Unfortunately, because the coding is done by replication, the rate of the resulting code is necessarily small.  A second line of work adds redundancy through true ``coding'' (eg, taking nontrivial linear combinations).  This allows for high-rate codes without much overhead in terms of the total computational load, but instead of focusing on bandwidth, this line of work has focused on minimizing the number $\tilde{k}$ of nodes that need to respond.
Several works in this second line have focused on linear functions, like matrix-vector multiplication~\cite{LLPPR18,DCG19} or Fourier transforms~\cite{YMA17b}; to the best of our knowledge, none of these have focused on download bandwidth beyond minimizing the number $\tilde{k}$ of workers that need to respond.
 
Our work provides a way to interpolate between these two lines of work.  That is, our work gives coded computation schemes for linear functions that \em both \em can have low download bandwidth \em and \em that can use non-replication-based coding to achieve a high rate.
In particular, Theorem~\ref{thm:informal} shows that we can use a rate $1 - \eps$ RS code, with bandwidth that scales like $n/\eps$, saving an $O(\log n)$ factor when $\eps$ is constant.  
As per Remarks~\ref{rem:nonlinear} and \ref{rem:prime} above, this approach can be used effectively to compute, say, $\ell_p$ norms over the reals, even though our Theorem~\ref{thm:informal} is stated for linear functions over extension fields.

We note that this is not directly comparable to prior work for coded computation of linear functions (eg, \cite{LLPPR18,DCG19,YMA17}) for two reasons.  First, those works have focused on computations with a larger output (eg, matrix-vector multiplication, where the output is a vector rather than a scalar), while our approach is most effective when the desired output is a scalar.  
Second, in much of the work on coded computation, the identities of the stragglers are not known to the other worker nodes.  In our approach, since the scheme may depend on the set $\mathcal{I}$ of failed nodes, the parameter server would have to broadcast this information, which may not be practical.  However, we note that the problem is still interesting even if there are no stragglers, simply to reduce download bandwidth (as in \cite{LMYA17}); or when the ``stragglers'' can be planned (for example to do load balancing between multiple tasks).

\paragraph{Homomorphic Secret Sharing.}  We have described the basic set-up for Homomorphic Secret Sharing (HSS) above.  Our scheme immediately gives a compact single-client HSS scheme for linear functions, by sharing a secret $\bx \in \F^k$ using a generalization of Shamir's scheme (as in \cite{FY92}) as follows.  Let $\bx \in \F^k$ be a secret.  Let $\tilde{k} > k$, so that $\tilde{k} + k < n$.  We encode $\bx$ with a systematic Reed-Solomon code, so that $x_i = f(\beta_i)$ for $i=1, \ldots, k$, where $f$ is a random polynomial of degree at most $\tilde{k}-1$ so that this is true, and where $\beta_1, \ldots, \beta_k \in \F$ are fixed evaluation points.  Then we distribute shares $f(\alpha_1), \ldots, f(\alpha_n)$ to the $n$ parties, where $\alpha_i \in \F\setminus \{\beta_1, \ldots, \beta_k\}$.  Now, any $\tilde{k}$ parties can recover the secret, while any $\tilde{k}-k$ learn nothing about it.  Theorem~\ref{thm:informal} (treating the evaluation points $\beta_1, \ldots, \beta_k$ as the unavailable nodes in $\cI$) ensures that as long as $\eps = 1-\tilde{k}/n$ and $\gamma = k/n$ are constants with $\gamma < \eps$, then each party can compute a small local share $g_i(f(\alpha_i))$, which can then be combined to recover a linear function $F(\bx)$.

As noted in Remarks~\ref{rem:nonlinear} and \ref{rem:prime} above, this approach can also be used for amortizing the computation of certain (possibly nonlinear) functions over prime fields. 

\subsection{Related Work}\label{sec:related}

First, we mention two works that are similar in flavor to ours in that the aim is to compute functions on data encoded with an error correcting code, although the models are quite different. The first of these is~\cite{CGdW13}, which studies the notion of \emph{error-correcting data structures}.  In that work, a vector $\bx$, thought of as a database, is encoded as a data structure $\mathcal{C}(\bx)$; as in our work, the goal is to efficiently compute some function (e.g., perform a membership query) on $\bx$ given access to $\mathcal{C}(\bx)$, possibly in the presence of noise.  However, that work differs from ours because (a) they consider query complexity (rather than bandwidth) as the notion of efficiency; and (b) the noise they consider is errors (rather than erasures).  Thus, in some sense, the work \cite{CGdW13} generalizes locally decodable codes in the same direction that we generalize regenerating codes.
The second work that is similar in flavor is
the recent work~\cite{LBAY21} on \emph{function correcting codes}.  In that work, a sender Alice sends a message $\bx$ over a noisy channel to a receiver Bob who is only interested in some function $f(\bx)$.  The main focus of that work is on the trade-off between the amount of noise in the channel and the rate of the code, given that Bob can recover $f(\bx)$.  This differentiates their problem from ours because they do not study any notion of efficiency (like bandwidth or query complexity) on Bob's end.

As mentioned above, notions related to Definition~\ref{def:repair} arise in a variety of contexts, including in regenerating codes, coded computation, and homomorphic secret sharing.  We survey related work in these areas below.  

\paragraph{Regenerating codes.}
The body of work most related to ours is that of regenerating codes.  Regenerating codes were introduced in \cite{dimakis} and have seen a huge amount of work since then.  The work most related to ours is the study of scalar\footnote{In the regenerating codes literature, a \em scalar \em MDS code is one that is linear over its alphabet, as opposed to a \em vector \em MDS code, which is linear over a smaller field.} MDS codes, including RS codes.  This was initiated by \cite{SPDC14}, and further developed in a line of work including \cite{GW17, TYB18}.  These works give \em repair schemes \em for RS codes, which can be seen as evaluation schemes for RS codes and for the class of functions $\mathcal{F} = \inset{ F_i: \bx \mapsto \cC(x)_i\suchthat i \in [n]}$.  The work~\cite{GW17} gives a characterization of repair schemes for MDS codes.  This characterization inspires our Definition~\ref{def:linearScheme} and Proposition~\ref{prop:linearIsEnough}, which gives a similar formulation for evaluation schemes for linear codes and classes of linear functions.  However, our framework for RS codes developed in Section~\ref{sec:framework} is quite different than the approach in \cite{GW17}.  In more detail, in \cite{GW17}, the goal is to choose dual codewords $\by^{(1)}, \ldots, \by^{(t)}$, so that they give rise to low-dimensional $\B$-subpsaces.  In contrast, our approach is to go the other way around: we first pick the low-dimensional $\B$-subspaces, and show how they give rise to appropriate dual codewords.

\paragraph{Coded computation.}  As mentioned above, there are two main lines of work in coded computation.  We refer the reader to \cite{CCsurvey} for a survey.  
One line of work has focused on coding for stragglers and has used ``true'' coding (in the sense that linear combinations of the original data are stored, rather than repeated blocks).  In our framework, stragglers correspond to the set $\cI$ of failures, the code maps some data $\bx$ to a codeword $\bc$ that is distributed to workers, and the goal is to compute some function $F(\bx)$ from computations $g_i(c_i)$ performed by the worker $i$ on their part of the encoded data $c_i$.  (We note that typically in these settings the symbols $x_j$ and $c_i$ are actually vector or matrix-valued, and the code is applied to each coordinate in parallel).  This line of work has considered linear functions like matrix-vector multiplication~\cite{LLPPR18,DCG19} or Fourier transformations~\cite{YMA17b}, as well as non-linear functions like matrix-matrix multiplication~\cite{YMA17} and computation of low-degree polynomials~\cite{YLRKSA19}.  The main focus has been on minimizing the number of workers required to complete their task before the desired function can be computed, as well as on analyzing when and how much this can speed up computation given stochastic models of stragglers.  However, to the best of our knowledge, this line of work has not considered the network bandwidth, which is what we consider here.

A second line of work has also focused on coding for stragglers, but has used replication-based coding.  That is, the data $\bx$ is separated into blocks, and these blocks are distributed to workers with repetition.  For example, worker $i$ might receive blocks 1 and 2, and worker $j$ might receive blocks 2 and 3.  This approach is especially common in the area of \em gradient coding \em \cite{gradient_coding, HASH18, RTTD20}, where the goal is to compute the function $F_{\btheta}(\bx) = \nabla L( \bx;\btheta )$ which is the gradient of a loss function at a current iterate $\btheta$.  In this set-up, again the main goal is to minimize the number of nodes that need to respond before the function can be computed, but some works like \cite{YA18} have also considered the download bandwidth.  Thus, the goal of \cite{YA18} is similar to ours, but the approach differs because (a) they are using a replication-based code, and in particular the rate must be low; but (b) their scheme does not depend on the identity of the stragglers, which ours does.
We note that there are several relaxations of the gradient coding problem, for example when the stragglers are random and/or the gradient only needs to be approximately computed~\cite{CPE17, RTTD20, LKAS18}.  Again, those works differ from ours because of the replication-based coding and the different model of stragglers.

A final line of
work, starting with \cite{LMYA17}, has focused on minimizing communication bandwidth, as we do here, but in a different setting.
That work considers computation in a general MapReduce framework.  In that work, the data is distributed before the Map phase, introducing redundancy via replication.  Then the data is shuffled before the Reduce phase; the goal is to reduce the amount of communication in the shuffle.  Finally, the Reduce phase occurs, and each node needs to compute the function that they are responsible for.  This can be viewed as a decentralized version of our setting where each node wants to compute a (different) function.  Key differences between that work and ours are that (a) the coding comes via replication, and (b) the goal is to be able to support generic computation in the MapReduce framework, rather than focusing on specific functions.

\paragraph{Homomorphic Secret Sharing.}
Homomorphic Secret Sharing was introduced in \cite{HSS0} and has been further explored in \cite{HSS2} and the references therein.  As noted above, a single-client compact HSS scheme is related to our definition of low-bandwidth function evaluation, where the code is given by the secret-sharing scheme.  The work \cite{HSS0} gave a two-party HSS scheme for any deterministic branching program that is cryptographically secure; this scheme has been optimized in \cite{HSS1}, and other works~\cite{BKS19,OSY21,RS21} have achieved similar results under different cryptographic assumptions.  The work \cite{HSS2} has studied the problem more generally, including under information-theoretic security, and provided lower bounds. While the setup of HSS is quite related to our work, most existing work on HSS is in a very different parameter regime.  For example, the two-party case studied in \cite{HSS0} corresponds to a code of length $n=2$.  
Additionally, since an MDS code provides an information-theoretically secure secret-sharing scheme, HSS is most related to our work under information-theoretic security.  However, most constructions that we are aware of for HSS have focused on cryptographic security.
One exception is the recent work~\cite{FIKW21}, which focuses on the download bandwidth of information-theoretic HSS.  However, that work focuses on \emph{multi}-client HSS, where the $k$ secrets in $\bx \in \F^k$ must be secret-shared independently of each other; in contrast, the application of our work sketched above is for \emph{single}-client HSS, where the $k$ secrets in $\bx$ may be shared jointly.

\subsection{Organization}
In Section~\ref{sec:tech} we set notation and give a brief overview of our approach.  In Section~\ref{sec:prelim} we introduce our framework for linear functions and linear codes.  In Section~\ref{sec:framework} we introduce our framework for RS codes in particular.  In Section~\ref{sec:proofs}, we instantiate our framework to prove Theorem~\ref{thm:main}, the more detailed version of Theorem~\ref{thm:informal} above.  Section~\ref{sec:conc} concludes with some open questions.

\section{Notation and Technical Overview}\label{sec:tech}
In this section we set some notation and give a quick technical overview of the main ideas in our work.
\subsection{Notation}
Throughout, we use $[n]$ to denote the set $\{1, 2, \ldots, n\}$.  
We use bold lowercase letters like $\bx$ to denote vectors, and bold uppercase letters like $\bG$ to denote matrices.  For a vector $\bx$, we use $x_i$ to denote the $i$'th coordinate of $x$.  We use $\bx|_{[i,j]}$ to denote the vector $(x_i, x_{i+1}, \ldots, x_{j})$.
For a polynomial $f(X) = \sum_i f_i X^i$, we define the \em degree set \em of $f$ to be
\[ \degset(f(X)) = \inset{ i \suchthat f_i \neq 0 }. \]

We always work over a field $\F = \F_Q$, where $Q = q^t$ and we will let $\B = \F_q$ be the subfield of $\F$ of size $q$.  With $Q = q^t$ as above, we will make use of the \em field trace \em of $\F_Q$ over $\F_q$, defined by
\[ \tr(X) = \sum_{i=0}^{t-1} X^{q^i}.\]
We note the following two facts about the field trace:
\begin{itemize}
\item The field trace is $\F_q$-linear and its image is contained in $\F_q$.
\item The field $\F = \F_Q$ is a vector space over the subfield $\B = \F_q$.  Given a basis $\zeta_1, \ldots, \zeta_t$ for $\F$ over $\B$, the traces $\tr(\zeta_1 \alpha), \ldots, \tr(\zeta_t \alpha)$ uniquely specify $\alpha \in \F$.
\end{itemize}

We consider \em linear \em codes $\cC: \F^k \to \F^n$.  Such a code can be represented by a full-rank \em generator matrix \em $\bG \in \F^{n \times k}$, so that $\cC(\bx) = \bG\bx$ for $\bx \in \F^k$.

We consider both $\F$-subspaces of $\F^n$ and $\B$-subspaces of $\F$ or $\F^n$.
To that end, we use $\spn$, $\dim$, and $\cdot^{\perp}$ (with no decoration) to refer to the span, dimension, and orthogonal complement over $\F$.  We use $\spn_{\B}$, $\dim_\B$ and $\cdot^{\perp_\B}$ (decorated with a ``$\B$'') to denote the span, dimension, and orthogonal complement over $\B$.  We define the orthogonal complement over $\B$ as follows.
For a $\B$-vector space $V \subseteq \F$, we define $V^{\perp_\B} := \inset{ x \in \F \suchthat \tr(xv) = 0 \forall v \in V }$.  
For a $\B$-vector space $\cV \subset \F^n$, we define
\[ \cV^{\perp_\B} := \inset{ \bx \in \F^n \suchthat \tr( \bx^T \bv ) = 0 \forall \bv \in V}. \]

\subsection{Technical Overview}
Our approach begins with a general linear-algebraic framework, similar to that from \cite{GW17} for renegerating codes.  
Let $\B$ be a subfield of $\F$, and let $\zeta_1, \ldots, \zeta_t$ be a basis for $\F$ over $\B$.
For a code $\cC:\F^k \to \F^n$, 
let $C = \cC(\F^k)$, so $C$ is a subspace of dimension $k$ in $\F^n$ that consists of all codewords.
We can associate an evaluation scheme with a sequence of $\B$-subspaces $V_1, \ldots, V_n \subset \F$, by demanding that node $j$ return enough information to evaluate $\tr( c_j \nu )$ for all $\nu \in V_j$.  Since $V_j$ is $\B$-linear, it suffices to send $b_j$ symbols from $\B$, where $b_j = \dim_\B(V_j)$.  When is this enough information to recover a linear function $F_\bp(\bx) = \bp^T \bx$?

In Definition~\ref{def:linearScheme}, we define a \em linear evaluation scheme \em as a sequence of $\B$-subspaces $V_1, \ldots, V_n \subset \F$ that has a nice relationship to $\cC$, and then we show that this nice relationship allows us to recover linear functions $F_\bp(\bx)$.  More precisely, 
let
\[ \cW = V_1^{\perp_\B} \times \ldots, \times V_n^{\perp_\B}. \]
We show in Section~\ref{sec:prelim} that if $\spn_\F\inparen{ C \cap \cW}$ 
has low dimension over $\F$, then there are many linear functions $F_\bp$ that can be recovered by the scheme derived from the subspaces $V_1, \ldots, V_n$. 
Thus, the goal becomes to find $\B$-subspaces $V_1, \ldots, V_n \subset \F$ so that $\spn_\F\inparen{ C \cap \cW }$ is low-dimensional over $\F$.  (Additionally, we need to keep track of \em which \em linear functions we can recover, but we will gloss over that in this overview).  Notice that $\cW$ is a $\B$-vector space, but \em not \em an $\F$-vector space.  Thus, it is not obvious how to get a handle on the dimension of this span.

In order to control the dimension of $\spn_\F\inparen{ C \cap \cW}$, we specialize to Reed-Solomon codes (rather than any linear code); this is where our analysis departs in similarity from \cite{GW17}.  We do this in Sections~\ref{sec:framework} and \ref{sec:proofs}.
Suppose that we choose $V_i = \spn_\B\inparen{v(\alpha_i)}$, where $v(X) \in \F[X]$ is some polynomial.  Then our goal becomes to show that 
\begin{equation}\label{eq:inf_want}
\inset{ g \in \F[X] \suchthat \deg(g) < k, \tr( g(\alpha_j) v(\alpha_j) ) = 0\ \ \forall j \in [n]} 
\end{equation}
lies in a low-dimensional $\F$-vector space.  Again, this is tricky because ``$\tr( g(\alpha_j) v(\alpha_j) ) = 0$'' is a $\B$-linear constraint, and we want $\F$-linear constraints.  We turn these $\B$-linear constraints into $\F$-linear constraints as follows.  Consider the unique polynomial $R(X)$ of degree at most $n-1$ so that 
\[ R(X) \equiv \tr( g(X)v(X) ) \mod p_A(X), \]
where $p_A(X) = \prod_{j=1}^n (X - \alpha_j)$.
Now, if $\tr(g(\alpha_j)v(\alpha_j)) = 0$ for all $j$, then $R$ vanishes everywhere and is thus identically zero.
The polynomial $R$ is a bit tricky to write down, but if the evaluation points are all of $\F$, then $p_A(X) = X^Q - X$, and in fact taking the residue of $R(X)$ modulo $p_A(X)$ is tractable.  Thus, our strategy is to expand out $R(X)$ and choose the coefficients of $v$ carefully so that the coefficient on some term $X^d$ is of the form $\sum_\ell v_{d-\ell} g_\ell$.  Since that coefficient must be zero---because $R(X)$ is identically zero---this gives us an $\F$-linear constraint on the polynomial $g$.  
If we get enough linearly independent $\F$-linear constraints this way, we can show that the space \eqref{eq:inf_want} lies in a low-dimension $\F$-vector space, which in turn will show that there are many $F_{\bp}$ that can be recovered by the scheme associated with $V_i = \spn_\B(v(\alpha_i))$.
(Again, in this overview we gloss over the fact that we actually want to know \emph{which} functions $F_{\bp}$ can be recovered this way:
by keeping track of exactly which linear constraints we get, we are able to design the polynomial $v(X)$ so that we can control this.)

The approach above is sufficient to design a scheme for codes of rate up to $1/2$, that doesn't tolerate any failures $\cI$.  As a warm-up, we present this result as Theorem~\ref{thm:ratehalf}.  In order to extend our result to get Theorem~\ref{thm:main}, the more detailed version of Theorem~\ref{thm:informal} above, we must choose several polynomials $v^{(1)}, \ldots, v^{(s)}$, increasing the bandwidth by a factor of $s$.  There are two main ideas here.  First, in order to make the rate larger than $1/2$ in the scheme from Theorem~\ref{thm:ratehalf}, we must restrict not only the coefficients of $v$ but also the coefficients of $g$.  This results in a scheme for a subset $\cF'$ of linear functions.  By repeating this several times, we are able to recover all of the linear functions.  Second, in order to handle failures in an arbitrary set $\cI$, we choose the polynomials $v^{(r)}(X)$ to additionally vanish on the set $\cI$.  Indeed, since the subspace $V_j$ given by a polynomial $v(X)$ is $V_j = \spn_\B(v(\alpha_j))$, if $v$ vanishes on $\cI$ then $V_j = \{0\}$ for all $j \in \cI$.  Thus, the dimension is zero, and the $j$'th node does not need to return any information. 

We give the framework for general linear codes, and explain why $\spn_\F\inparen{C \cap \cW}$ is important, in Section~\ref{sec:prelim}.  We develop our ``pick nice polynomials $v(X)$'' framework in Section~\ref{sec:framework}.  Finally, we instantiate our framework for a full-length RS code and analyze it in Section~\ref{sec:proofs}.

\section{Framework for linear functions and any linear code}\label{sec:prelim}

For the rest of the paper, we focus on the special case where $\cC$ is a linear code, and where $\mathcal{F}$ is a set of linear functions.  In this case, evaluation schemes for $\mathcal{F}$ and $\cC$ can arise from a simple linear-algebraic condition, defined next.
\begin{definition}[Linear Evaluation Schemes]\label{def:linearScheme}
Let $\B$ be a subfield of $\F$, and let $\zeta_1, \ldots, \zeta_t \in \F$ be a basis for $\F$ over $\B$.
Let $\cC:\F^k \to \F^n$ be a linear code, and let $\bG \in \F^{n \times k}$ be a generator matrix for $\cC$.
Let $\bp \in \F^k$ and let $\bw \in \F^n$ be any vector so that $\bp = \bG^T \bw$.  (Note that such a vector exists since $\bG$ has full column-rank).

Suppose that $V_1, \ldots, V_n \subset \F$ are $\B$-subspaces, so that $\dim_{\B}(V_j) = b_j$.
Let $\cV = V_1 \times \cdots \times V_n \subset \F^n$.
We say that $(V_1, \ldots, V_n)$ provide a \emph{linear evaluation scheme for $\bp$ and $\mathcal{C}$} (with respect to $\{\zeta_1, \ldots, \zeta_t\}$) if for all $i \in [t]$,
\[ \zeta_i \bw \in \cC(\F^k)^\perp + \cV \]
The \emph{bandwidth} of the scheme is $\left( \sum_{j=1}^n b_j \right)\lceil \log|\B| \rceil$.
Further, we say that $(V_1, \ldots, V_n)$ \emph{tolerates failures in $\cI$} where $\cI := \inset{ j \in [n] \suchthat V_j = \{0\} }$.

\medskip

For $\cP \subseteq \F^k$, we 
 say that a map $\vphi:\cP \to (2^\F)^n$ 
provides a linear evaluation scheme for $\cP$ and $\cC$ if $\vphi(\bp) = (V_1, \ldots, V_n)$ provides a linear evaluation scheme for $\bp$ and $\cC$ for all $\bp \in \cP$.
For a set $\cI \subset [n]$,
we say that $\vphi$ tolerates failures in $\cI$ if, for all $\bp \in \cP$,  $\vphi(\bp)$ tolerates failures in $\cI$.
\end{definition}

The following proposition explains why a linear evaluation scheme $\vphi$ indeed gives us an evaluation scheme for a set of linear functions.
\begin{proposition}\label{prop:linearIsEnough}
Suppose that $\vphi$ provides a linear evaluation scheme for $\mathcal{P}$ and $\mathcal{C}$, with bandwidth $b$.  Then there is an evaluation scheme for the class of functions 
\[ \mathcal{F} = \{ F_{\bp}:\bx \mapsto \bx^T \bp \suchthat \bp \in \mathcal{P} \} \]  
and $\cC$ with bandwidth $b$.  Moreover, for $\cI \subset [n]$, this evaluation scheme tolerates failures in $\cI$ if $\vphi$ does. 
\end{proposition}

\begin{proof}
Suppose that $\vphi$ forms a linear evaluation scheme for $\mathcal{P}$ and $\mathcal{C}$.  Let $\bp \in \mathcal{P}$, and let $\bw \in \F^n$ be as in the theorem statement.
Let $(V_1, \ldots, V_n) = \varphi(\bp)$ be the linear evaluation scheme for $\bp$.
For each $j \in [n]$, let $\beta^{(j)}_1, \ldots, \beta^{(j)}_{b_j} \in V_j$ be a basis for $V_j$ over $\B$.  We will construct functions $g_1, \ldots, g_n$ and $G$ as in Definition~\ref{def:repair} that will allow us to reconstruct $F(\bx) = \bp^T \bx$.

Fix $\bx \in \F^k$ and let $\bc = \bG \bx \in \cC(\F^k)$ be the corresponding codeword.
First, we observe that
\[ \bc^T \bw = \bx^T \bG^T \bw = \bx^T \bp = F(\bx). \]
Thus, we focus on recovering $\bc^T \bw$.

We define the functions $g_j:\F \to \B^{b_j}$ by
\[ g_j(x) := ( \tr(c_j \beta_1^{(j)}), \tr(c_j, \beta_2^{(j)}), \ldots, \tr(c_j, \beta^{(j)}_{b_j}) ). \]
Now by the definition of a linear evaluation scheme, for all $i$, 
$\zeta_i \bw \in \cC(\F^k)^\perp + \cV$.  This implies that 
there are dual codewords $\bz^{(1)}, \ldots, \bz^{(t)} \in \cC(\F^k)^\perp$ so that for all $j \in [n]$ and all $i \in [t]$,
\[ \zeta_i w_j - z_j^{(i)} \in V_j. \]
In order to define the repair function $G$, we observe that for all $i \in [t]$, we have
\[ \zeta_i \sum_j w_j c_j = \sum_j c_j(\zeta_i w_j - z_j^{(i)}). \]
This is because $\sum_j c_j z_j^{(i)} = 0$, since $\bz^{(i)} \in \cC(\F^k)^\perp$ and $\bc \in \cC(\F^k)$.
In particular, this implies that for all $i \in [t]$,
\begin{equation}\label{eq:1}
 \tr\left( \zeta_i \sum_j w_j c_j \right) = \sum_j \tr( c_j(\zeta_i w_j - z_j^{(i)})). 
\end{equation}
Since $\zeta_i w_j - z_j^{(i)} \in V_j$ for all $j \in [n], i \in [t]$, we can write
\[ \zeta_i w_j - z_j^{(i)} = \sum_{\ell=1}^{b_j} a^{(i,j)}_\ell \beta^{(i)}_\ell \]
for some coefficients $a^{(i,j)}_\ell \in \B$.  Thus,
\begin{equation}\label{eq:2}
 \tr(c_j(\zeta_i w_j - z_j^{(i)})) = \sum_{\ell=1}^{b_j} a^{(i,j)}_\ell \tr( c_j \beta^{(i)}_\ell), 
\end{equation}
where above we have used the linearity of the trace.
Now, we can define the function $G$ to be the output of the following algorithm:
\begin{itemize}
	\item Input: $g_j(c_j) = ( \tr(c_j \beta_1^{(j)}), \tr(c_j, \beta_2^{(j)}), \ldots, \tr(c_j, \beta^{(j)}_{b_j}) )$ for all $j$.
	\item For each $i \in [t], j \in [n]$, use \eqref{eq:2} to recover $\tr(c_j(\zeta_i w_j - z_j^{(i)}))$ from the input.
	\item Use \eqref{eq:1} to recover $\tr\left( \zeta_i \sum_j w_j c_j \right)$ for all $i\in [t]$.
	\item Since $\zeta_1, \ldots, \zeta_t$ form a basis for $\F$ over $\B$, this is sufficient to recover $\sum_j w_j c_j = F(\bx)$.  Return $F(\bx)$.
\end{itemize}
Thus, we have a function evaluation protocol for $\cC$.  To finish the proof, we observe that the amount of information sent is $\sum_{j=1}^n b_j$ symbols from $\B$, so the total bandwidth is $\inparen{\sum_{j=1}^n b_j }\lceil\log|\B|\rceil$ bits.

Finally, observe that if $V_j = \{0\}$, then $b_j = 0$ and the scheme above does not need to contact symbol $i$, so this scheme tolerates failures in $\cI$.
\end{proof}

In the next lemma, we reformulate the condition in Definition~\ref{def:linearScheme} in a way that will be helpful going forward.
\begin{lemma}\label{lem:perp_char}
Let $\cV = V_1 \times V_2 \times \cdots \times V_n \subset \F^n$, where each $V_i$ is a $\B$-subspace of $\F$.  
Let $\cW = W_1 \times W_2 \times \cdots \times W_n \subset \F^n$, where $W_i = V_i^{\perp_{\B}}$.
Let $\cC:\F^k \to \F^n$ be a linear code.  Let $\zeta_1, \ldots, \zeta_t$ be a basis for $\F$ over $\B$.  Then for any $\bw \in \F^n$, 
\[ \zeta_i \bw \in \cC(\F^k)^\perp + \cV \ \ \ \forall i \in [t] \]
if and only if
\[ \bw \in \inparen{ \spn_{\F}(\cC(\F^k) \cap \cW) }^{\perp}. \]
\end{lemma}
\begin{proof}
Let $C = \cC(\F^k) \subset \F^n$.
We have
\begin{align}
\bw \in \inparen{ \spn_{\F}(C \cap \cW) }^\perp &
\Leftrightarrow \forall \by \in C \cap \cW, \bw^T \by = 0 \notag \\
&\Leftrightarrow \forall \by \in C \cap \cW, \forall i \in [t], \tr( (\zeta_i \bw)^T \by ) = 0 \notag \\
&\Leftrightarrow \forall i \in [t], \zeta_i \bw \in (C \cap \cW)^{\perp_\B} \notag \\
&\Leftrightarrow \forall i \in [t], \zeta_i \bw \in C^{\perp_\B} + \cW^{\perp_\B},  \label{eq:4}
\end{align}
using the fact that for any vector spaces $A,B$, we have $(A \cap B)^\perp = A^\perp + B^\perp$.  
Now, we observe that $C^{\perp_\B} = C^\perp$.  Indeed, 
\begin{align*}
 C^{\perp_\B} &= \inset{ \bx \suchthat \tr( \bc^T \bx ) = 0 \forall \bc \in C } \\
&= \inset{ \bx \suchthat \tr( \zeta_i \bc^T \bx ) = 0 \forall \bc \in C, \forall i \in [t] } \\
&= \inset{ \bx \suchthat \bc^T \bx = 0 \forall \bc \in C } = C^\perp.
\end{align*}
Above, we used the fact that since $\cC$ is linear, $C = \zeta_i C$ for all $i$.
Further, we observe that $\cW^{\perp_\B} = \cV$ by defintion.
Thus, from \eqref{eq:4}, we conclude that
\[ \bw \in \inparen{ \spn_{\F}(C \cap \cW)}^\perp \Leftrightarrow
\forall i \in [t], \zeta_i \bw \in C^{\perp} + \cV, \]
which is what we wanted to show.
\end{proof}

\section{Framework for linear functions and RS codes}\label{sec:framework}

The framework in Section~\ref{sec:prelim} was valid for any linear code $\cC$.  Now, we specialize to Reed-Solomon codes in order to leverage this characterization.  We begin with a few definitions that will be useful for our framework.

\begin{definition}\label{def:sigma}
Let $A = (\alpha_1, \ldots, \alpha_n)$, so that $\alpha_i \in \F$ are distinct.  
Define 
\[ p_A(X) = \prod_{j=1}^n (X - \alpha_j). \]
For a non-negative integer $j$ and for $i \in \{0, 1, \ldots, t-1\}$, define $\sigma_i(j) \subset \mathbb{Z}$ to be 
\[ \sigma_i(j) = \degset\inparen{\overline{X^{j q^i}}}, \]
where $\overline{X^{j q^i}}$ is the unique polynomial of degree at most $n-1$ so that
\[ \overline{ X^{jq^i} } \equiv X^{jq^i} \mod p_A(X). \] 
\end{definition}
We note that $\sigma_i$ depends on the choice of $A$, but we suppress this dependence in the notation for readability.

\begin{remark}\label{rem:sigmaZero}
For any $j < n$, we have $\sigma_0(j) = \degset\inparen{ \overline{X^j}} = \degset\inparen{X^j} = \inset{j}$.
\end{remark}

\begin{remark}\label{rem:specialSigma}
While for general $A$, $\sigma_i$ may be quite complicated, for some sets $A$ it is relatively simple.  For example, if $A = \F$, then $p_A(X) = X^Q - X$, and
\[ \sigma_i(j) = \{ jq^i \modstar Q-1 \}, \]
where
\[ x \modstar Q-1 := \begin{cases} y \in \{1, \ldots, Q-1\} \text{ so that } y \equiv x \mod Q-1 & \text{ if } x \neq 0 \\ 0 & \text{ if } x = 0 \end{cases}.\]
In particular, if we write $x \in \{0, 1, \ldots, Q-1\}$ in base-$q$ as
\[ x = \sum_{b=0}^{t-1} x_b q^b \]
for $x_b \in \{0,\ldots, q-1\}$,
then 
\[\sigma_i(x) = \left\{ \sum_{b=0}^{t-1} x_{b} q^{b + i \mod t-1} \right\} \]
is a circular shift of this expansion.
\end{remark}

\begin{definition}\label{def:good}
Let $0 < k \leq n$ and consider the Reed-Solomon code $\cC$ of dimension $k$ with evaluation points $A = (\alpha_1, \ldots, \alpha_n)$ over $\F$.
Let $\jmin, \jmax, d$ be positive integers so that $\jmin \leq \jmax$ and $\jmin < d$.
We say that $(\jmin, \jmax, d)$ is \emph{good} for $\cC$ if all of the following hold: 
\begin{enumerate}
\item $d < n$ and $\jmax + k - 1 < n$; 
\item for all $i \in  \{1, \ldots, t-1\}$, $d \not\in \bigcup_{j=\jmin}^{\jmax + k - 1} \sigma_i(j)$; and
\item $d \in \bigcup_{j=\jmin}^{\jmax + k - 1} \sigma_0(j)$,
\end{enumerate}
where above $\sigma_i$ is defined as in Definition~\ref{def:sigma} with respect to $A$.
Given some $d, \jmin, \jmax$, we define $\ellmin$ and $\ellmax$ by
\begin{equation}\label{eq:ells} 
\ellmin = \max\{ 0, d - \jmax \} \qquad \text{and} \qquad \ellmax = \min\{ k-1, d- \jmin \}.
\end{equation}
\end{definition}

\begin{definition}\label{def:consistent}
Fix $(d, \jmin, \jmax)$, and let $\ellmin,\ellmax$ be as in \eqref{eq:ells}. 
Let $\bp \in \F^k$, so that $\supp(\bp) \subseteq [\ellmin, \ellmax]$.  We say that a polynomial
\[ v(X) = \sum_{j=\jmin}^{\jmax} v_j X^j \]
is \emph{consistent} with $\bp$ (with respect to $(d, \jmin, \jmax)$), if $v_j = p_{d-j}$ whenever $d-j \in [0,k-1]$.
\end{definition}

Notice that, for any $\bp$ as in Definition~\ref{def:consistent}, there is some polynomial $v(X)$ consistent with $\bp$, given by
$v(X) = \sum_{j=\max\{d - k + 1, \jmin\}}^{\min\{d, \jmax\}} p_{d-j} X^j.$

With these definitions, we have the following lemma.
\begin{lemma}\label{lem:goodparity}
Let $0 < k \leq n$ and consider the Reed-Solomon code $\cC$ of dimension $k$ with evaluation points $A = (\alpha_1, \ldots, \alpha_n)$ over $\F$.
Suppose that $(\jmin, \jmax, d)$ is good for $\cC$, and let $\ellmin, \ellmax$ be as in Definition~\ref{def:good}.
Then for all $\bp \in \F^k$ so that $\supp(\bp) \subseteq [\ellmin, \ellmax]$, 
and for all $v(X)$ consistent with $\bp$,
there exist $\B$-subspaces $V_1, \ldots, V_n \subset \F$ with $\dim_{\B}(V_j) \leq 1$ for all $j$, so that the following holds:

Let $\bg \in \F^k$ and let 
 $g(X) = \sum_{\ell=0}^{k-1} g_\ell X^\ell$.  Suppose that $g(\alpha_j) \in V_j^{\perp_\B}$ for all $j \in [n]$.  Then $\bg^T \bp = 0$.

Further,  $V_j = \{0\}$ for all $j \in \cI$, where
$\cI = \inset{j \in [n] \suchthat v(\alpha_j) = 0}$.  (Notice that $\cI$ depends on both $\bp$ and the choice of $v(X)$).
\end{lemma}

Before we prove Lemma~\ref{lem:goodparity}, we show how to use it to obtain a linear evaluation scheme for $\cC$.  The following Theorem is our main framework theorem for RS codes.

\begin{theorem}[Main Framework Theorem]\label{thm:framework}
Let $\cC$ be a Reed-Solomon code.
Suppose that $(\jmin, \jmax, d)$ is good for $\cC$.  Let $\lmin, \lmax$ be as in \eqref{eq:ells}, and let
\[ \cP \subseteq \inset{ \bp \in \F^k \suchthat \supp(\bp) \subseteq [\lmin, \lmax] }. \]
Then there is a linear evaluation scheme $\vphi$ for $\cP$ and $\cC$ with bandwidth at most $n \lceil \log|\B| \rceil$.

\medskip
Further, for any collection 
\[ \inset{ v_{\bp}(X) \suchthat \bp \in \cP } \]
so that for each $\bp \in \cP$, $v_{\bp}(X)$ is consistent with $\bp$, there exists a scheme $\vphi$
that tolerates failures in 
\[ \cI := \inset{j \in [n] \suchthat v_\bp(\alpha_j) = 0 \ \ \forall \bp \in \cP},\] 
 with bandwidth is at most
\[ (n - |\cI|) \lceil \log|\B|\rceil. \]

\end{theorem}
\begin{proof}[Proof of Theorem~\ref{thm:framework}, assuming Lemma~\ref{lem:goodparity}]
We prove the ``Further'' statement, since it implies that first statement.  (Indeed, we may take $v_\bp(X)$ to be any polynomial consistent with $\bp$).

Suppose that $(\jmin, \jmax, d)$ is good for $\cC$.  
Let $\bp \in \cP$, and suppose that $v_{\bp}(X)$ is consistent with $\bp$.
Let $\bG \in \F^{n \times k}$ be a generator matrix for $\cC$ and let $\bw$ be such that $\bp = \bG^T \bw$, as in Definition~\ref{def:linearScheme}.
Let $V_1, \ldots, V_n$ be the $\B$-subspaces guaranteed for $\bp$ and $v_{\bp}(X)$ by Lemma~\ref{lem:goodparity}, and define $W_j = V_j^{\perp_\B}$ for $j = 1, \ldots, n$.  Let $\cW = W_1 \times \cdots \times W_n$ and let $\cV = V_1 \times \cdots \times V_n$.

The guarantee of Lemma~\ref{lem:goodparity} implies that $\bc^T \bw = 0$ for all $\bc \in \cC(\F^k) \cap \cW$, so $\bw \in (\cC(\F^k) \cap \cW)^\perp$.  By Lemma~\ref{lem:perp_char}, this implies that 
\[ \zeta_i \bw \in \cC(\F^k)^\perp + \cV \qquad \forall i \in [t], \]
where $\zeta_1, \ldots, \zeta_t$ is a basis for $\F$ over $\B$.
Therefore from Definition~\ref{def:linearScheme}, $(V_1, \ldots, V_n)$ is a linear scheme for $\bp$ and $\cC$, and the map $\vphi$ that maps $\bp$ to $(V_1, \ldots, V_n)$ as above is a linear scheme for $\cP$ and $\cC$.

Further, Lemma~\ref{lem:goodparity} implies that for all $\bp \in \cP$, if $\vphi(\bp) = (V_1, \ldots, V_n)$ then $V_j = \{0\}$ for all $j \in \cI$.  Thus, $\vphi$ tolerates failures in $\cI$.

Finally, we observe that the bandwidth of the scheme is $\log|\B|$ times the number of $V_j$ so that $V_j \neq \{0\}$, which is at most $(n - |\cI|)\log|\B|$.
\end{proof}

Finally, we prove Lemma~\ref{lem:goodparity}.
\begin{proof}[Proof of Lemma~\ref{lem:goodparity}]
Suppose that $\cC$ is an RS code as in the statement of the lemma, so we have evaluation points $A = (\alpha_1, \ldots, \alpha_n)$.  Suppose that $(\jmin, \jmax, d)$ is good for $\cC$.  Choose $\bp \in \F^k$ such that $\supp(\bp) \in [\ellmin, \ellmax]$, and suppose that $v(X) = \sum_{j=\jmin}^{\jmax} v_j X^j$ is consistent with $\bp$.
Define
$ V_j := \spn_{\B}(v(\alpha_j)) $
for $j \in [n]$.  Notice that $\dim_\B(V_j) \leq 1$, as desired, and further that $V_j = \{0\}$ if $v(\alpha_j) = 0$.
Now suppose that $\bg \in \F^k$ so that $g(X) = \sum_{\ell=0}^{k-1} g_\ell X^\ell$ has $g(\alpha_j) \in V_j^{\perp_\B}$.
We wish to show that $\bg^T \bp = 0$.

From the definition of $V_j$ and the assumption that $g(\alpha_j) \in V_j^{\perp_\B}$ for all $j \in [n]$, we have
\[ \tr(v(\alpha_j)g(\alpha_j)) = 0 \]
for all $j$.  Consider the unique polynomial $R(X)$ of degree at most $n-1$ so that
\[ R(X) \equiv \tr(v(X)g(X)) \mod p_A(X), \]
where $p_A(X)$ is as in Definition~\ref{def:sigma}.
Thus, $R(\alpha_i) = 0$ for all $i \in [n]$.  Since $\deg(R) \leq n-1$, this implies that $R(X) \equiv 0$ is identically zero.
Consider the coefficient of $X^d$ in $R(X)$.  On the one hand, this is zero.  On the other hand, we can compute
\begin{align*}
 R(X) &= \sum_{i=0}^{t-1} \inparen{ \sum_{j=\jmin}^{\jmax} v_j X^j }^{q^i}\inparen{ \sum_{\ell=0}^{k-1} g_\ell X^\ell}^{q^i} \\
&= \sum_{i=0}^{t-1} \sum_{j, \ell} v_j^{q^i} g_\ell^{q^i} X^{ q^i(\ell + j) }
\end{align*}
Thus, we have
\begin{align}
0=
\text{(coefficient of $X^d$ in $R(X)$)} &= \sum_{i=0}^{t-1} \sum_{j, \ell : d \in \sigma_i(\ell + j) } c_{\ell+j, d, i} v_j^{q^i} g_\ell^{q^i} ,\label{eq:coeff}
\end{align}
where $c_{r,d,i} \in \F$ are the coefficients that arise when we write
\[ \overline{X^{rq^i}} = \sum_{d \in \sigma_i(r)} c_{r,d,i} X^d. \]
(Above, as in Definition~\ref{def:sigma}, $\overline{X^{rq^i}}$ refers to the residue modulo $p_A(X)$).
Since $(d, \jmin, \jmax)$ is good, Item 2 of Definition~\ref{def:good} says that
 for all $i \neq 0$, and for all $r \in [\jmin, \jmax + k - 1]$, $d \not\in \sigma_i(r)$.
Since $\ell + j \in [\jmin, \jmax + k - 1]$, this implies that
the inner sum on the right hand side of \eqref{eq:coeff} is empty if $i \neq 0$.
Therefore, we have
\begin{align}
0 &= 
\text{(coefficient of $X^d$ in $R(X)$)} \notag\\
&= \sum_{j, \ell : d \in \sigma_0(\ell + j) } c_{\ell+j, d, 0} v_j g_\ell \notag \\ 
&= \sum_{j = \jmin}^{\jmax} \sum_{\ell=0}^{k-1} \ind{\ell + j = d} v_j g_\ell \label{ln:1}\\
&=\sum_{j=\max\{d-k+1, \jmin\}}^{\min\{d, \jmax\}} \sum_{\ell=0}^{k-1} \ind{\ell + j = d} v_j g_\ell \label{ln:2}\\
&= \sum_{\ell=\ellmin}^{\ellmax} p_\ell g_\ell \label{ln:3} \\
&= \sum_{\ell = 0}^{k-1} p_\ell g_\ell. \label{ln:4}
\end{align}
Above, we have used in \eqref{ln:1}
the fact that $\sigma_0(\ell + j) = \{\ell + j\}$ (as per Remark~\ref{rem:sigmaZero}, using the assumption that $\ell + j \leq k-1+j_{\max} < n$ as per Definition~\ref{def:good}); and the fact that $c_{d,d,0} = 1$ since we have $\overline{X^{d}} = X^d$ (using the assumption that $d < n$).  In \eqref{ln:2}, we have used the fact that for $j \in [\jmin, d-k] \cup [d+1, \jmax]$, $\ind{\ell+j = d} = 0$.  In \eqref{ln:3}, we have used the definition \eqref{eq:ells} of $\ellmin$ and $\ellmax$.  And in \eqref{ln:4}, we have used the fact that $\supp(\bp) \subseteq [\ellmin, \ellmax]$.

This shows that $\bp^T \bg = 0$, which completes the proof.
\end{proof}

\section{Proof of main theorem}\label{sec:proofs}

We begin with a warm-up that already gives good schemes for RS codes of rates approaching $1/2$.

\begin{theorem}\label{thm:ratehalf}
Let $Q = q^t$, for some $t \geq 2$ and some prime power $q$.
Suppose that $k \leq Q\left( \frac{1}{q} \lflor \frac{q}{2} \rflor \inparen{ 1 - \frac{1}{q}} \right)$.  Let $\cC$ be the Reed-Solomon code of dimension $k$ and length $n = Q$ over $\F = \F_Q$.  Let $\cF$ be the class of all linear functions from $\F^k$ to $\F$:
\[ \mathcal{F} = \inset{ \left(F_{\by} : \bx \mapsto \bx^T\by\right) \suchthat \by \in \F^k }. \]
Then there is an evaluation scheme for $\cF$ and $\cC$ with bandwidth $n \lceil \log_2 q \rceil.$
\end{theorem}

Notice that the rate of the RS code $\cC$ in Theorem~\ref{thm:ratehalf} can be as large as 
\[ \frac{k}{n} = \frac{1}{q} \lflor\frac{q}{2}\rflor (1 - 1/q) \geq \frac{1}{2} - \frac{3}{2q}, \]
which approaches $1/2$ as $q$ grows.  We note that for $q=2$, the rate of $\cC$ is $1/4$.

\begin{proof}[Proof of Theorem~\ref{thm:ratehalf}]
We will use Theorem~\ref{thm:framework} to show that there is a linear scheme for $\cP = \F^k$.  Then Proposition~\ref{prop:linearIsEnough} will imply the theorem.

Choose 
\begin{align*}
d &= \lflor \frac{q}{2} \rflor q^{t-1} \\
\jmin &= \lflor \frac{q}{2} \rflor q^{t-2} + 1 \\
\jmax &= Q - k.
\end{align*}
We claim that $(d, \jmin, \jmax)$ is good for $\cC$.  We check the three items in Definition~\ref{def:good}: 
\begin{enumerate}
\item Since $n = Q = q^t$, and using the choice of $d$ above, we have $d < n$.  We also have $\jmax + k - 1 = Q-1 < n$.
\item As per Remark~\ref{rem:specialSigma}, for this full-length RS code we have $\sigma_i(j) = \{jq^i \modstar Q-1\}$.
Thus, the second item in Definition~\ref{def:good} is equivalent\footnote{Note here that for all $i\geq0$, $q^i$ is a unit of $\mathbb{Z}/(Q-1)\mathbb{Z}$, and in particular $q^t \modstar Q-1 = Q \modstar Q-1 = 1$.} to showing that for all $i = 1, \ldots, t-1$,
\begin{align*} d       &\neq q^i j \modstar Q-1 \qquad \forall j\in [\jmin, \jmax+k-1]\\
   dq^t \modstar Q-1 & \neq q^ij \modstar Q-1 \qquad \forall j\in [\jmin, \jmax+k-1]\\
   dq^{t-i} \modstar Q-1 & \neq j \modstar Q-1 \qquad \forall j\in [\jmin, \jmax+k-1]\\
  d q^{t-i} \modstar Q-1 & \not\in [ \jmin, \jmax + k - 1 ].
  \end{align*}
Plugging in the definitions of $d$, $\jmin$ and $\jmax$, this is the same as showing that for all $i = 1, \ldots, t-1$,
\[ \lflor \frac{q}{2} \rflor q^{t-i-1} \not\in \left[ \lflor \frac{q}{2} \rflor q^{t-2} + 1, Q-1 \right]. \]
This is true, because for all $i = 1, \ldots, t-1$, we have
\[ 0 <\lflor \frac{q}{2} \rflor q^{t-i-1} < \lflor \frac{q}{2} \rflor q^{t-2} + 1. \]
\item Finally, using the fact that $\sigma_0(j) = \{j\}$ for all $j \in [\jmin, \jmax + k - 1]$, the third item is equivalent to showing that $d \in [\jmin, \jmax + k - 1]$, or that
\[ \lflor \frac{q}{2} \rflor q^{t-1} \in \left[ \lflor \frac{q}{2} \rflor q^{t-2} + 1 , Q-1 \right], \]
which is true.
\end{enumerate}
Thus, $(d, \jmin, \jmax)$ is good for $\cC$.  
Now we compute $\lmin, \lmax$ as in Theorem~\ref{thm:framework}.  We have
\begin{align*}
\lmin&= \max\{0,  d - \jmax\} \\
&= \max\inset{ 0, \lflor \frac{q}{2} \rflor q^{t-1} - q^t + k } \\
&= 0,
\end{align*}
using the fact that $Q = q^t$ and 
\[ k \leq Q\inparen{ \frac{1}{q} \lflor \frac{q}{2} \rflor (1 - 1/q) } \leq Q\frac{1}{q} \lflor \frac{q}{2} \rflor \leq Q\inparen{ 1 - \frac{1}{q} \lflor \frac{q}{2}\rflor}. \]
We also have
\begin{align*}
\lmax &= \min\{k-1, d - \jmin\} \\
&= \min \inset{ k-1, \lflor \frac{q}{2} \rflor q^{t-1} - \lflor \frac{q}{2} \rflor q^{t-2} - 1 } \\
&= k-1,
\end{align*}
using the fact that $k \leq Q\inparen{ \frac{1}{q} \lflor \frac{q}{2} \rflor (1 - 1/q) }$.
Therefore we have
\[ \inset{ \bp \in \F^k \suchthat \supp(\bp) \subseteq [\ellmin, \ellmax] } = \F^k. \]

By Theorem~\ref{thm:framework} and the fact that $(d, \jmin, \jmax)$ is good for $\cC$, we conclude that there is a linear evaluation scheme $\varphi$ for $\cP = \F^k$, 
and $\cC$, with bandwidth $n \lflor \log q \rflor$, which is what we wanted to show.
\end{proof}

The reason that Theorem~\ref{thm:ratehalf} has rate limited by $1/2$ is that if we were to take $k$ to be larger, the interval $[\lmin, \lmax]$ would not be all of $[0,k-1]$.  In the next theorem, we modify the construction in Theorem~\ref{thm:ratehalf} to give a constant number of schemes like the one in Theorem~\ref{thm:ratehalf}, each of which covers a small interval, but which together cover all of $[0,k-1]$.  Thus, we can increase the rate of the code to approach $1$, at the cost of increasing the bandwidth by a constant factor.  While we are at it, we give ourselves enough freedom in order to choose the schemes so that they can tolerate failures in any set $\cI$ that is not too large.

\begin{theorem}\label{thm:main}
Let $Q = q^t$, for some $t \geq 2$ and prime power $q$.
Let $\eps, \gamma > 0$.
Let
\[ \delta \geq \gamma + \frac{1}{q}, \]
and suppose that $\eps > \delta$ and that $(\eps - \delta)q$ is an integer.
Suppose that $k \leq Q(1 - \eps)$, and let $\cC$ be the Reed-Solomon code of dimension $k$ and length $n = Q$ over $\F = \F_Q$.  Let $\cF$ be the class of all linear functions from $\F^k$ to $\F$:
\[ \mathcal{F} = \inset{ \left(f : \bx \mapsto \bx^T\by\right) \suchthat \by \in \F^k }. \]
Let $\cI \subset [n]$ be any set of size $|\cI| < \gamma n$.
\medskip

Then there is an evaluation scheme $\Phi$ for $\cF$ and $\cC$ that tolerates failures in $\cI$, and that has bandwidth at most
\[ (n - |\cI|) \cdot \inparen{\frac{1 }{ \eps - \delta}} \cdot   \lceil \log_2 q \rceil.\]

\end{theorem}
\begin{remark}
For constant $\gamma$, the requirements on $\eps, \gamma, \delta$ may be satisfied with a choice of $\eps = \gamma + \Theta(1/q)$.
Thus, as $q$ grows, $\eps$ may approach $\gamma$.
This means that the trade-off between the rate of the code ($1 - \eps$) and the fraction of failures tolerated ($\gamma$) approaches the Singleton bound, which is optimal (regardless of bandwidth).
\end{remark}

\begin{remark}
We have chosen to present Theorem~\ref{thm:main} as it applies to the full-length Reed-Solomon code of length $n=Q$.
However, the scheme can also be used for shorter codes with $n<Q$, as long as $Q-n \leq \gamma Q$.
This is because the scheme tolerates failures of up to $\gamma Q$ nodes, and we may instead imagine these nodes never existed in the first place.
In total, the number of failed or nonexistent nodes can be at most $\gamma Q$.
\end{remark}

\begin{proof}[Proof of Theorem~\ref{thm:main}]
Define $s$ to be the largest integer so that 
\[ s < \frac{1}{\eps - \delta}. \]
where $\delta$ is as in the theorem statement.
Before we proceed, we record the following useful claim:
\begin{claim}\label{claim:useful}
With $\delta,\eps$ as in the theorem statement, we have
\[ \delta \geq \frac{1- \eps}{q-1} + \frac{ \gamma}{1 - 1/q}. \]
\end{claim}
\begin{proof}
We have
\begin{align*}
\gamma + \frac{1}{q} &= \frac{1 -1/q}{q-1} + \gamma \\
&= \frac{ 1 - 1/q - \gamma }{q-1} + \frac{ q\gamma}{q-1} \\
&\geq \frac{ 1 - \eps }{q-1} + \frac{\gamma}{1 - 1/q},
\end{align*}
in the last line using the assumptions that 
\[ \eps \geq \delta \geq \gamma + \frac{1}{q}. \]
\end{proof}

For $r = 1, \ldots, s$, we will define a evaluation scheme $\Phi^{(r)}$ that tolerates failures in $\cI$.  Each of these evaluation schemes will only be able to recover linear functions with support in some window, but together the $\Phi^{(r)}$ will form an evaluation scheme for all of $\cF$.
We begin with the following claim.
\begin{claim}\label{claim:lr}
For $r = 1, \ldots, s$, there is a choice of $(\dr, \jminr, \jmaxr)$ 
so that:
\begin{enumerate}
\item $(\dr, \jminr, \jmaxr)$ is good for $\cC$ for all $1 \leq r \leq s$;
\item $d^{(1)} - \jmax^{(1)} \leq 0$;
\item $d^{(s)} - \jmin^{(s)} \geq k - 1 + Q\gamma$;
\item for all $1 \leq r < s$,
\[ \inparen{ \dr - \jminr } - \inparen{ d^{(r+1)} - \jmax^{(r+1)} } \geq Q\gamma - 1. \]
\end{enumerate}
\end{claim}
\begin{proof}
For $r = 1, \ldots, s$, define
\[ y^{(r)} = (\eps-\delta)qr.\]
Note that since $r \leq s < \frac{1}{\eps - \delta}$, we have $(\eps - \delta)r < 1$, and so $y^{(s)} < q$. Further, by our assumption that $(\eps - \delta)q \in \mathbb{Z}$, $y^{(r)}$ is an integer.
Define
\begin{align*}
 d^{(r)} &= y^{(r)} q^{t-1} \\
\jminr &= \yr q^{t-2} + 1 \\
\jmaxr &= Q - k.
\end{align*}
(Notice that these choices are reminiscent of the choices in the proof of Theorem~\ref{thm:ratehalf}).
First, we establish that each $(\dr, \jminr, \jmaxr)$ is good for $\cC$.  We check the three conditions in Definition~\ref{def:good}:
\begin{enumerate}
\item  For all $q \leq r \leq s$, we have
$\dr = y^{(r)} q^{t-1}$.  As noted above, $y^{(s)} < q$, and so we have
$\dr < q^t = Q = n$. 
Similarly we have $\jmaxr + k - 1 = Q-1 < n$.
\item As in the proof of Theorem~\ref{thm:ratehalf}, it suffices to show that 
\[ q^{t-i}\dr \modstar Q-1 \not\in [\jminr, \jmaxr + k - 1] \]
for all $i = 1, \ldots, t-1$.  
This is true since 
for all such $i$, we have
\[ 0 < q^{t-i} \dr \modstar Q-1 \leq \yr q^{t-2} < \jminr \]
using the definition of $\jminr$.
\item It suffices to show that $\dr \in [\jminr, \jmaxr + k - 1]$,which is equivalent to  
$\yr q^{t-1} \in [\yr q^{t-2} + 1, Q-1]$, which is true. 
\end{enumerate}
This establishes the first point of the claim.  

For the second point, we observe that
\begin{align*}
d^{(1)} - \jmax^{(1)} &=  y^{(1)} q^{t-1} - q^t + k \\
&= (\eps - \delta) Q - Q + Q(1 - \eps) \\
&= -\delta Q \leq 0.
\end{align*}

For the third point, we observe that
\begin{align*}
d^{(s)} - \jmin^{(s)} &=  y^{(s)} (q^{t-1} - q^{t-2}) - 1 \\
&= (\eps - \delta) sQ  (1 - 1/q) - 1 \\
&\geq (\eps - \delta) \inparen{ \frac{1}{\eps - \delta} - 1 } Q(1 -1/q) - 1 \\
&= (1 - \eps + \delta)Q (1 - 1/q) - 1,
\end{align*}
using the fact that $s \geq \frac{1}{\eps - \delta} - 1$.
In order for this to be at least $k - 1 + Q\gamma = Q(1 - \eps + \gamma) - 1$, we need
\begin{align*}
(1 - \eps + \delta)(1 - 1/q) &\geq 1 - \eps + \gamma \\
\delta(1 - 1/q) - \frac{1-\eps}{q} &\geq \gamma \\
\delta &\geq \frac{ 1-\eps}{q-1} + \frac{\gamma}{1 -1/q}, 
\end{align*}
which is indeed satisfied by our choice of $\delta$, by Claim~\ref{claim:useful}.
This establishes the third point.

Finally, for the fourth point, we compute
\begin{align*}
\inparen{ \dr - \jminr } - \inparen{ d^{(r+1)} - \jmax^{(r+1)}} &= \yr (q^{t-1} - q^{t-2}) -1 - y^{(r+1)} q^{t-1} + q^t - k \\
&= Qr(\eps - \delta)(1 - 1/q) - 1 - Q(r + 1)(\eps - \delta) + \eps Q \\
&= Q(\delta - r(\eps - \delta)/q ) - 1 \\
&\geq Q( \delta - s(\eps - \delta)/q ) - 1\\
&\geq Q\inparen{ \delta - \inparen{ \frac{ 1}{\eps - \delta}}\inparen{ \frac{ \eps - \delta}{q} }} - 1 \\
&= Q(\delta - 1/q) - 1 \\
&\geq Q\gamma - 1
\end{align*}
using the fact that $s \leq \frac{1}{\eps - \delta}$ in the third-to-last line, and using our assumption that $\delta\geq \gamma+1/q$ in the final line.
This establishes the last point, and proves the claim.

\end{proof}

\begin{claim}\label{claim:ps}
Let $(\dr, \jminr, \jmaxr)$ be as in Claim~\ref{claim:lr}, and let $\lminr = \max\{ 0, \dr - \jmaxr \} $ and $ \lmaxr = \min\{ k-1, \dr- \jminr \}$.
For any $\bp \in \F^k$, there is a sequence $\bp^{(1)}, \bp^{(2)}, \ldots, \bp^{(s)} \in \F^k$ so that:
\begin{enumerate}
\item For each $1 \leq r \leq s$, we have
\[\supp(\bp^{(r)}) \subseteq [\lminr, \lmaxr].\]
\item 
For each $1 \leq r \leq s$, there exists a polynomial $v^{(r)}(X)$ that is consistent with $\bp^{(r)}$, such that $v^{(r)}(\alpha_i) = 0$ for all $i \in \cI$.
\item We have $\sum_{r=1}^s \bp^{(r)} = \bp.$
\end{enumerate}
\end{claim}
\begin{proof}
We prove the claim by induction.  Suppose inductively that we have $\bp^{(1)}, \ldots, \bp^{(r-1)}$ and consistent polynomials $v^{(1)}(X), \ldots, v^{(r-1)}(X)$, so that $v^{(1)}, \ldots, v^{(r-1)}$ all vanish on $\cI$ and so that 
\begin{equation}\label{eq:IH} 
\left.\left(\sum_{j=1}^{r-1} \bp^{(j)}\right)\right|_{[0, \lminr-1]} = \bp|_{[0,\lminr-1]}. 
\end{equation}
(For the base case, we define $\bp^{(0)} = \mathbf{0}$, using Claim~\ref{claim:lr}, Item 2, to establish that that $d^{(1)} - \jmax^{(1)} \leq 0$ and hence $\lmin^{(1)} =0$, and taking the convention that $[0,-1]=\emptyset$.)

Now, given $\bp^{(1)}, \ldots, \bp^{(r-1)}$, we define $\bp^{(r)}$ for $r < s$ as follows.
First, we define
\[ \left.\bp^{(r)}\right|_{[\lminr, \ellmin^{(r+1)}-1]} = \left. \bp \right|_{[\lminr, \ellmin^{(r+1)}-1]} - \sum_{j=1}^{r-1} \left. \bp^{(j)}\right|_{[\lminr, \ellmin^{(r+1)}-1]}. \]
Observe that, by induction, this implies that
\[ \left.\left(\sum_{j=1}^{r} \bp^{(j)}\right)\right|_{[0, \lmin^{(r+1)} - 1]} = \bp|_{[0,\lmin^{(r+1)} - 1]}. \]
Now, we need to define $\left.\bp^{(r)}\right|_{[\lmin^{(r+1)}, \ellmax^{(r)}]}$ and $v^{(r)}(X)$.  
Write 
\[ v^{(r)}(X) = \sum_{j = \jminr}^{\jmaxr} v_j^{(r)} X^j, \]
where we must define the $v_j^{(r)}$.  Whenever $\dr - j \in [\lminr, \lmin^{(r+1)}-1]$, we define
\[ v_j^{(r)} := p^{(r)}_{\dr - j}, \]
noting that for such $j$, $p^{(r)}_{\dr - j}$ is already inductively defined.
Next, we choose the remaining coefficients $v_j^{(r)}$ in order to make $v^{(r)}(X)$ vanish on $\cI$.  This is possible because the number of free coefficients is at least $|\cI|$.  Indeed, we have already set all of the coefficients for $j \geq \dr - \lmin^{(r+1)} + 1$, and this leaves free all of the coefficients from $j = \jminr$ to $j = \dr - \lmin^{(r+1)}$. 
The number of these is
\begin{align*}
\dr - \lmin^{(r+1)} - \jminr + 1&= 
\inparen{\dr - \jminr} - \ellmin^{(r+1)} + 1 \\
&\geq \inparen{ \dr - \jmin^{(r)}} - \inparen{ d^{(r+1)} - \jmax^{(r+1)}} \\
&\geq \gamma Q - 1 \geq |\cI|,
\end{align*}
where in the last line we have used Claim~\ref{claim:lr}, Item 4.
Thus, we may choose the remaining coefficients $v_j^{(r)}$ so that $v^{(r)}$ vanishes on $\cI$.  Then we define
\[ p^{(r)}_\ell := v_{\dr - \ell}^{(r)} \]
for all $\ell \in [\ellmin^{(r+1)}, \ellmax^{(r)}]$, noting that these are all defined since
 $\jminr \leq \dr - \ellmax^{(r)}$ and we have defined the coefficients $v_j^{(r)}$ all the way down to $j = \jmin^{(r)}$.

Finally, we note that by construction, $\supp(\bp^{(r)}) \subseteq [\lminr, \lmaxr]$, and that $\bp^{(r)}$ is consistent with $v^{(r)}$, so items 1 and 2 of Claim~\ref{claim:ps} are satisfied for $r$.

Now we have constructed $\bp^{(r)}$ and $v^{(r)}$ that satisfy the inductive hypothesis \eqref{eq:IH} for $r$.  By induction, we can construct these for all $r = 1, \ldots, s-1$.  

To conclude, we will define $\bp^{(s)}$ and $v^{(s)}$ slightly differently.  We choose $\bp^{(s)}$ to have support contained in $[\lmin^{(s)}, k-1]$ so that
\[ \bp^{(s)}|_{[\lmin^{(s)}, k-1]} = \bp|_{[\lmin^{(s)}, k-1]} - \sum_{r=1}^{s-1} \bp^{(r)}|_{[\lmin^{(s)}, k-1]}. \]
Then, as before, we define the corresponding coefficients of $v^{(s)}$ so that $v^{(s)}$ is consistent with $\bp^{(s)}$.  To do this, we must define
\[ v^{(s)}_j := p^{(s)}_{d^{(s)}-j} \]
for all $j \in [\jmin^{(s)}, \jmax^{(s)}]$ so that $d^{(s)} -j \in [0, k-1]$.
By Claim~\ref{claim:lr}, Item 3, $d^{(s)} - \jmin^{(s)} \geq k-1 + \gamma Q$, so there are at least $\gamma Q$ values of $j \in [\jmin^{(s)}, \jmax^{(s)}] \setminus [d^{(s)}-k+1, d^{(s)}]$.  Thus, as above, we may use the fact that $|\cI| < \gamma Q$ and choose coefficients $v_j^{(s)}$ for $j$ in this set so that $v^{(s)}(X)$ vanishes on $\cI$.  

Notice that Claim~\ref{claim:lr}, Item 3, also implies that $\ell_{max}^{(s)} =k-1$, and so 
we have that $\supp(\bp^{(s)}) \subseteq [\ellmin^{(s)}, \ellmax^{(s)}]$ by construction.  By construction we also have that $v^{(s)}(X)$ is consistent with $\bp^{(s)}$, and also that $v^{(s)}(X)$ vanishes on $\cI$.  Thus points 1 and 2 in Claim~\ref{claim:ps} are satisfied for $\bp^{(s)}, v^{(s)}$ as well.

Finally, using \eqref{eq:IH} for $r=s-1$ and our choices for $\bp^{(s)}$, we have
\begin{align*}
 \sum_{r=1}^s \bp^{(r)} &= \left.\left(\sum_{r=1}^{s} \bp^{(r)}\right)\right|_{[0, \lmin^{(s)}-1]}  + \left.\left(\sum_{r=1}^{s-1} \bp^{(r)}\right)\right|_{[\lmin^{(s)}, k-1]} + \bp^{(s)}|_{[\lmin^{(s)}, k-1]} \\
&= \bp|_{[0, \lmin^{(s)}-1]} + \bp|_{[\lmin^{(s)}, k-1]} \\
&= \bp,
\end{align*}
as desired.
This finishes the proof of the claim.
\end{proof}

Finally, we describe the scheme $\Phi$ that the theorem guarantees.  
For $r = 1, \ldots, s$, let $\cP^{(r)}$ be the set of vectors $\bp^{(r)} \in \F^k$ that can arise from Claim~\ref{claim:ps}.  (That is, Claim~\ref{claim:ps} says that for all $\bp \in \F^k$, there exists $\bp^{(1)},\ldots, \bp^{(s)}$ with particular properties; for each $\bp$ pick an arbitrary such sequence and include $\bp^{(r)}$ in $\cP^{(r)}$.)

Let $\vphi^{(r)}$ be the linear scheme guaranteed by Theorem~\ref{thm:framework} for $\cP^{(r)}$, so $\vphi^{(r)}$ has bandwidth at most $(n-|\cI|) \lceil \log q \rceil$.  
(Here, we are using the fact from Claim~\ref{claim:lr} that $(\dr, \jminr, \jmaxr)$ are good for $\cC$).  The second point in Claim~\ref{claim:ps} ensures that each $\vphi^{(r)}$ tolerates errors in $\cI$.  

Now, let $\Phi^{(r)}$ be the evaluation scheme guaranteed by Proposition~\ref{prop:linearIsEnough}.  By that proposition, each of these schemes has bandwidth at most $(n-|\cI|) \lceil \log q \rceil$, and also tolerates errors in $\cI$.

Finally, we define $\Phi$ as follows.  
Given $\bp \in \F^k$, define $F_{\bp}:\F^k \to \F$ by $F_{\bp}(\bff) = \bff^T \bp$.

\noindent \textbf{Scheme $\Phi$:} Suppose that the original data was $\bff \in \F^k$.
Given input $F_{\bp} \in \cF$:
\begin{itemize}
\item Let $\bp^{(1)}, \ldots, \bp^{(s)}$ be as in Claim~\ref{claim:ps}.
\item For each $r = 1, \ldots, s$, use $\Phi^{(r)}(\bp^{(r)})$
 to download $(n - |\cI|)\lceil \log q \rceil$ bits and recover $\bff^T \bp^{(r)}$.
\item Return $\bff^T \bp = \sum_r \bff^T\bp^{(r)} $.
\end{itemize}

The correctness follows from Claim~\ref{claim:ps}, and the bandwidth is at most 
\[ (n - |\cI|) \cdot s \cdot \lceil \log q \rceil. \]
Plugging in the definition of $s$ proves the theorem. 

\end{proof}

\section{Conclusion}\label{sec:conc}
In this paper we considered low-bandwidth function evaluation on encoded data.  Special cases of this problem appear throughout computer science, engineering and cryptography, and we believe that it is valuable to study this problem in generality.  We kick off this agenda by studying the problem for \em general linear functions \em and for \em Reed-Solomon codes, \em arguably among the most natural classes of functions and codes.  However, we hope that this is just the tip of the iceberg.  We conclude with several questions left open by this work.
\begin{enumerate}
\item Can we develop low-bandwidth evaluation schemes for other classes of functions?  (Beyond those mentioned in Remark~\ref{rem:nonlinear} that are implied by our results?)  Low-degree polynomials are perhaps the next most natural class.  
\item Can we develop low-bandwidth evaluation schemes for linear functions, for general linear codes?  The first part of our framework (in Section~\ref{sec:prelim}) applies to general linear codes, but the second part (Section~\ref{sec:framework}) and our main theorem applies only for RS codes.
\item Can we extend our scheme to work in different parameter regimes?  In particular, our scheme works with full-length RS codes over extension fields.  Work from regenerating codes has shown how to use RS codes as regenerating codes in very different parameter regimes, for example when $t$ is very large~\cite{TYB18} or over prime fields~\cite{CT21}.  Could these approaches be adapted to low-bandwidth function evaluation?
\end{enumerate}

\section*{Acknowledgements}
We thank Yuval Ishai for helpful conversations, and in particular for suggesting the approach in Remark~\ref{rem:prime}.
We thank Ravi Vakil for helpful conversations.
\bibliographystyle{halpha-2}
\bibliography{refs.bib}

\appendix
\section{Bandwidth lower bound}\label{app:lb}
In this appendix, we observe a few lower bounds on the bandwidth required for linear evaluation schemes.

We say that an MDS code with generator matrix $\bG \in \F^{n \times k}$ is \em maximal \em if there is no way to add a row to extend $\bG$ to a matrix $\bG' \in \F^{(n+1) \times k}$ that is also MDS.  The following fact is standard.
\begin{fact} Reed-Solomon codes are \em not \em maximal. \label{fact:rsnonmax} \end{fact}
\begin{proof}
Let $\bG \in \F^{n \times k}$ be the generator matrix for an RS code over $\F = \F_Q$.  If $n < Q$, then we may extend $\bG$ to be the generator matrix of an RS code with a larger set of evaluation points.  If $n = Q$, then we may extend $\bG$ to be the generator matrix for the \em doubly-extended RS code, \em by adding the row $(0,0, \ldots, 0,1)$.
\end{proof}

First we observe that if $\cC$ is a non-maximal MDS code, then any 
lower bound that holds for repairing a single symbol in a regenerating code also holds for 
linear evaluation schemes, for the class of all linear functions.
Indeed, suppose that $\bg \in \F^k$ is the row that we would add to $\bG$ to get a new MDS matrix $\bG'$.  Then consider the linear function $F(\bx) = \bg^T \bx$.  Now $\bG'$ corresponds to an MDS code $\cC':\F^k \to \F^{n+1}$.  If $bc' = \cC'(\bx)$, then $c'_{n+1} = F(\bx)$, and so any linear evaluation scheme that will evaluate $F$ given access to $\bc = \bc'|_{[n]}$ is also a repair scheme for $\cC'$.

We can collect various lower bounds, including the cut-set bound of \cite{dimakis} and the lower bound for MDS codes of \cite{GW17} (see also \cite{DM17} for a more precise version), to obtain the following observation.
\begin{observation}\label{obs:nonmaxlb}
Let $\F = \F_Q$ where $Q = q^t$.
Suppose that $\cC:\F^k \to \F^n$ is an MDS code that is not maximally MDS.  (In particular, $\cC$ may be any RS code, by Fact~\ref{fact:rsnonmax}).
Let $\cF$ be the class of all linear functions $F:\F^k \to \F$.
Then any linear evaluation scheme for $\cF$ and for $\cC$ over the base field $\mathbb{B} = \F_q$ must have bandwidth at least
\[ b \geq \max \inset{ k + t - 1, t \inparen{\frac{ n }{n-k+1}}, n \log_q\inparen{\frac{n}{n-k+1} }}. \]
\end{observation}

We note that when $\cC$ is a maximal MDS code (or not an MDS code at all), the logic above does not go through.  Indeed, if $\cC$ is a maximal MDS code, then definitionally any linear function of $\bx$ can be computed by looking at fewer than $k$ nodes.
However, we are able to mimic the proof of the lower bound from \cite{GW17} for any linear code.

\begin{proposition}\label{prop:lb}
Let $\F = \F_Q$, for $Q = q^t$.
Let $\cC: \F^k \to \F^n$ be a linear code with generator matrix $\bG \in \F^{n \times k}$, let $\bp \in  \F^k$, and suppose that $(V_1, \ldots, V_n)$ forms a linear evaluation scheme for $\bp$ and $\cC$, over the base field $\B = \F_q$.    Let $C = \cC(\F^k)$.
Choose any $\bw \in \F^n$ so that $\bG^T \bw = \bp$.
Let 
\[ d^* = \min_{ \by \in C^\perp} \Delta( \bw , \by), \]
where $\Delta$ denotes Hamming distance.  Then the bandwidth $b$ of $(V_1, \ldots, V_n)$  satisfies
\[ b \geq n \log_q \inparen{ \frac{1}{1 - (1 - 1/Q) d^*/n} }. \]
\end{proposition}
Before we prove Proposition~\ref{prop:lb}, we observe a corollary for MDS codes. 

\begin{corollary}\label{cor:lbrs}
Let $\F =\F_Q$ where $Q = q^t$.
Let $\cC: \F^k \to \F^n$ be an MDS code with $n > k+1$.  If $\vphi$ is a linear evaluation scheme for $\cC$ and for $\cP = \F^k$ over the base field $\B = \F_q$, then $\vphi$ has bandwidth
\[ b \geq n \log_q\inparen{\frac{n}{n-k+3}}. \]
\end{corollary}
\begin{proof}
Suppose that $\cC: \F^k \to \F^n$ is an MDS code.
We show that the quantity $d^*$ in Proposition~\ref{prop:lb} is at least $k-1$.
First, we 
use the fact that
the \em covering radius \em $r(\cC)$ of any MDS code $\cC$ satisfies
\begin{equation}\label{eq:covering}
r(\cC) := \max_{\bw \in \F^n} \min_{\by \in C} \Delta( \bw, \by ) \geq n - k - 1
\end{equation}
(see~\cite{BGP14}).
Next, we observe that the quantity
$\min_{ \by \in C^\perp} \Delta( \bw^T \bp, \by)$
does not depend on the choice of $\bw$.  Indeed, suppose that $\bz \in \F^n$ also satisfies $\bG^T \bz = \bp$.
Then $\bG^T (\bz-\bw) = 0$, so $\bz = \bw + \bu$ for some $\bu \in C^\perp$.  
But then
\[ \min_{\by \in C^\perp} \Delta( \bw, \by ) = \min_{\by \in C^\perp} \Delta(\bw, \by - \bu) = \min_{\by \in C^\perp} \Delta(\bz, \by), \]
where in the first equality we have changed the order of summation.  In particular, as $\bp$ ranges over all of $\F^k$, we may choose $\bw$ to range over all of $\F^n$.  Applying \eqref{eq:covering} for $C^\perp$ (which we may do as the dual of an MDS code is again MDS), we see that by choosing an appropriate $\bp \in \F^k$, we may take
\[ d^* = \max_{\bw \in \F^n} \min_{\by \in C^\perp} \Delta( \bw, \by ) = r(C^\perp) \geq k - 1. \]
(Above, we have switched the role of ``$k$'' and ``$n-k$'' from \eqref{eq:covering} since the dimension of $C^\perp$ is $n-k$).
This gives the required bound on $d^*$.  The  corollary follows after plugging $d^* \geq k-1$ into Proposition~\ref{prop:lb}, and using the fact that $k \leq n \leq 2Q$ for any MDS code with $n > k+1$.\footnote{This fact follows from (a) $q \geq n - k + 1$ together with (b) $q \geq k+1$ if $n > k+1$.  See, e.g., \cite{ball} Lemmas 1.2 and 1.3 for (a) and (b) respectively.}
\end{proof}

\begin{proof}[Proof of Proposition~\ref{prop:lb}]
The proof follows similarly to the lower bound for regenerating codes proved in \cite{GW17}.
Let $\bp, \bw$ be as in the statement of the proposition and suppose that $(V_1, \ldots, V_n)$ forms a linear evaluation scheme for $\bp$ and $\cC$, over the base field $\B = \F_q$, with bandwidth at most $b$. Suppose that $b_j = \dim(V_j)$, so $\sum_{j=1}^n b_j \leq b$.  Let $C = \cC(\F^k)$.  By definition, this means that for all $i$,
\[ \zeta_i \bw \in C^T + \cV, \]
where $\cV = V_1 \times V_2 \times \cdots \times V_n$.
Let $\by^{(1)}, \ldots, \by^{(t)} \in C^\perp$ and $\bv^{(1)}, \ldots, \bv^{(t)} \in \cV$ be such that 
\[ \zeta_i \bw = \by^{(i)} + \bv^{(i)} \]
for all $i = 1, \ldots, t$.  Now consider a random vector $\bb \in \B^t$.  From the above, we have
\begin{align*}
\sum_i b_i \inparen{\zeta_i \bw -  \by^{(i)}} &= \sum_i b_i \bv^{(i)} \\
\zeta_{\bb} \bw - \by_{\bb} & \in \cV,
\end{align*}
where above we are defining $\zeta_{\bb} = \sum_i b_i \zeta_i$ and $\by_{\bb} = \sum_i b_i \by^{(i)}$.
Thus, the $j$'th symbol of $\zeta_\bb \bw - \by_{\bb}$ is in $V_j$, and we claim that it is in fact uniform on a $\B$-subspace of $V_j$.  Indeed, this is because by definition we are choosing said $j$'th symbol to be a random $\B$-linear combination of the elements $\zeta_i w_j - y^{(i)}_j$ for $i=1,\ldots,t$, so it will be uniform on the $\B$-subspace spanned by those elements.  Therefore, for each $j = 1, \ldots, n$, the probability that the $j$'th symbol of $\zeta_\bb \bw - \by_{\bb}$ is zero is at least $1/|V_j| = q^{-b_j}$.

This implies that the expected number of zeros in $\zeta_\bb \bw - \by_{\bb}$ is at least
\[ \mathbb{E}\inabs{ \inset{ j \in [n] \suchthat \zeta_\bb w_j = (y_{\bb})_j }} \geq \sum_{j=1}^n q^{-b_j} \geq n q^{-b/n}, \]
where above we have used the fact that $\sum_j q^{-b_j}$ is minimized (subject to $\sum_j b_j \leq b$) when all the $b_j$ are the same, and equal to $b/n$.  Thus, 
\[ \mathbb{E} \Delta(\zeta_{\bb} \bw , \by_{\bb} ) \leq n (1 - q^{-b/n}). \]
If $\bb = \bm{0}$, the distance is zero, so
\[ \mathbb{E} \inbrak{\Delta(\zeta_\bb \bw, \by_\bb) \,\mid \, \bb \neq \bm{0} } = \frac{q^t}{q^t - 1} \mathbb{E}\Delta(\zeta_{\bb} \bw , \by_{\bb} ) \leq \frac{q^t}{q^t - 1} n (1 - q^{-b/n}). \]
Thus, there exists a $\bb \neq \bm{0}$ so that 
\[ \Delta( \bw, \zeta^{-1}_\bb \by_\bb ) \leq \frac{q^t}{q^t - 1} n (1 - q^{-b/n}). \]
Above, we have used the fact that $\zeta_1, \ldots, \zeta_t$ form a basis, so $\zeta_{\bb}$ is nonzero if $\bb$ is nonzero.
As in the proof of Corollary~\ref{cor:lbrs}, we note that the definition of $d^*$ does not depend on the choice of $\bw$ so that $\bG^T \bw = \bp$.  Thus, by the definition of $d^*$, we have that
\[ d^* \leq \Delta( \bw, \zeta^{-1}_\bb \by_\bb ) \leq \frac{q^t}{q^t - 1} n (1 - q^{-b/n}). \]
Solving for $b$, we see that
\[ b \geq n \log_q \inparen{ \frac{1}{1 - (1-1/Q)d^*/n}}, \]
as desired.
\end{proof}

\end{document}